\documentclass[journal]{new-aiaa}

\usepackage{amsmath}

\usepackage{amssymb}
\usepackage{doi} % doi links
\usepackage{mathtools} % underbracket
\usepackage{mathrsfs} % mathscr
\usepackage{glossaries} % gls
\usepackage{siunitx} % si
\usepackage{bm} % bm to bold math symbols
\usepackage{amsthm} % proof env
\usepackage{multirow} % table cells with multiple rows

\usepackage{tikz}
\usepackage{pgfplots}
\usepgfplotslibrary{groupplots}
\usepackage[linesnumbered,ruled,vlined]{algorithm2e} % algorithm
\usetikzlibrary{external}
\graphicspath{figures}
\graphicspath{ {./img/} }

\makeatletter
\newcommand*{\centerfloat}{%
  \parindent \z@
  \leftskip \z@ \@plus 1fil \@minus \textwidth
  \rightskip\leftskip
  \parfillskip \z@skip}
\makeatother

\usepackage[scr=boondoxo,scrscaled=1.05]{mathalfa}

\newcommand{\Iset}{\mathscr{I}}
\newcommand{\covfun}{\Sigma}
\renewcommand{\t}{^{\mbox{\tiny \sf T}}}

\newcommand{\addt}{\mbox{\tiny \sf T}}

\newcommand{\inv}{^{- 1}}

\renewcommand{\d}{\mathrm{d}}
\renewcommand{\Pr}{\mathbb{P}}
\newcommand{\E}{\mathbb{E}}
\renewcommand{\Re}{\mathbb{R}}
\newcommand{\Ncal}{\mathcal{N}}

\newcommand{\Amat}{\bm{A}}
\newcommand{\Bmat}{\bm{B}}
\newcommand{\Gmat}{\bm{G}}
\newcommand{\Qmat}{\bm{Q}}
\newcommand{\Rmat}{\bm{R}}
\newcommand{\Lmat}{\bm{L}}

\newcommand{\Kmat}{\bm{K}}
\newcommand{\Smat}{\bm{S}}
\newcommand{\PXmat}{\bm{P_X}}
\newcommand{\PUmat}{\bm{P_U}}

\DeclareMathOperator{\tr}{tr}
\newcommand{\Cov}{\mathrm{Cov}} % Has less whitespace after than math operator

\DeclareMathOperator{\ncdf}{cdfn}

\newtheorem{remark}{Remark}
\newtheorem{theorem}{Theorem}[section]

\newtheorem{prop}[theorem]{Proposition}

\title{Chance-Constrained Covariance Steering in a\\Gaussian Random Field\\via Successive Convex Programming}

\author{Jack Ridderhof%
\footnote{PhD Candidate, School of Aerospace Engineering.}
and Panagiotis Tsiotras%
\footnote{David and Andrew Lewis Chair and Professor, School of Aerospace Engineering, and Institute for Robotics and Intelligent Machines.}}
\affil{Georgia Institute of Technology, Atlanta, GA, 30332}

\begin{document}

\maketitle

\begin{abstract}
	The problem of optimizing affine feedback laws that explicitly steer the mean and covariance of an uncertain system state in the presence of a Gaussian random field is considered.
	Spatially-dependent disturbances are successively approximated with respect to a nominal trajectory by a sequence of jointly Gaussian random vectors.
	Sequential updates to the nominal control inputs are computed via convex optimization that includes the effect of affine state feedback, the perturbing effects of spatial disturbances, and chance constraints on the closed-loop state and control.
	The developed method is applied to solve for an affine feedback law to minimize the 99\textsuperscript{th} percentile of $\Delta v$ required to complete an aerocapture mission around a planet with a randomly disturbed atmosphere.
\end{abstract}

\section{Introduction}
\label{sec:intro}

Random disturbances acting on autonomous systems are often spatially dependent.
Examples include variations in atmospheric properties \cite{Justus2002,Ridderhof2020aeroconf}, underwater currents \cite{Lee2019currents}, and gravitational fields \cite{DeMars2015}.
The uncertain nature of these disturbances leads the system state to be a random variable with statistics determined by the system dynamics, the probabilistic structure of the disturbances, and the system control law. 
While the system dynamics and the probabilistic structure of the disturbances are fixed, it is possible to design the feedback control to desirably shape the evolution of the system probability distribution.
Indeed, for the case in which the state is Gaussian distributed, steering the state covariance by optimizing over the feedback gains has been formulated as a convex program \cite{Okamoto2018csl}.
However, the stochastic control literature is primarily concerned with systems affected by temporal disturbances, such as Brownian motion, rather than spatial disturbances.
The aim of this paper is to bridge the gap between the treatment of spatial and temporal disturbances for feedback control design, and to solve for affine feedback laws that explicitly steer the mean and covariance of the system state, subject to chance constraints, while the system is affected by spatially-dependent uncertainty.

In this paper, we model spatial uncertainty as a Gaussian random field (GRF), which can be thought of as a generalization of the Gaussian distribution to function spaces \cite{Rasmussen2006gp}.
Similarly to a Gaussian random vector, a GRF is fully characterized by a mean and a covariance function.
For any finite number of inputs (e.g., a set of position vectors), the values of the GRF are jointly Gaussian distributed with mean and covariance determined by evaluating the mean and covariance functions at the input points.

GRF models have been widely applied in the fields of spatial analysis \cite{Krige1951,Camps-Valls2016}, machine learning \cite{Rasmussen2006gp}, robotics \cite{Anderson2015,Mukadam2016,Kreuzer2018mppi}, and state estimation \cite{DeMars2015,Olson2017}.
For many applications, including the aforementioned references, GRF models are primarily used for either regression or for characterization of a 
yet-to-be-explored unknown environment.
This paper, in contrast, is concerned with using a GRF to characterize disturbances to be handled by feedback control, similar to how disturbances are treated in classical stochastic control, such as Linear Quadratic Gaussian (LQG) control.

We take as a motivating example the problem of aerocapture, which is an orbital aeroassist maneuver where a spacecraft uses a planet's atmosphere to decelerate from a hyperbolic orbit to a captured elliptical orbit around the planet \cite{Lockwood2004neptune}.
During aerocapture, the spacecraft must fly through the atmosphere of another planet, which may not be well characterized, at orbital velocity.
Descending into the lower atmosphere results in (exponentially) higher density and thus more drag, which increases the effectiveness of the maneuver --- but the perturbing effect of density variations is also much greater at the lower attitudes.
Furthermore, assuming that atmospheric density variations depend, at least partly, on the altitude, the density variations seen by the vehicle following periapsis are correlated to previously encountered variations \cite{Ridderhof2020aeroconf,Ridderhof2021entry}.
While atmospheric density uncertainty is a major driver of performance, no methodology currently exists to \textit{explicitly} treat atmospheric uncertainty for guidance and control optimization.
Rather, the state-of-the-art closed-loop predictor-corrector guidance successively treats the atmosphere as being equal to an onboard current best estimate, and performs a deterministic optimization.
The resulting guidance performance is evaluated through Monte Carlo analysis that includes spatial density variations, and guidance parameters are tuned based on the Monte Carlo results \cite{Lu2015aerocapture,Matz2020,Justus2002}.
%In Section~\ref{sec:example_areocap}, we apply the method developed in this paper to aerocapture guidance with the planet's atmospheric density given as a GRF.

This paper takes a sequential optimization approach to solve for both a feedforward (nominal) control and corresponding state feedback gains.
We begin with a nominal trajectory that does not account for uncertainty, and which takes the GRF to be equal to its mean value.
This trajectory may be the solution to a deterministic optimal control problem, for example.
Assuming that, in the presence of uncertainty, the trajectory will not deviate too far from its nominal value, the perturbing effect of the GRF can be approximated by the statistics of the GRF evaluated along the nominal trajectory.
In other words, the nominal trajectory serves as a mapping between time and space, which is used to reduce the spatial GRF to a temporal process.
Trajectory disturbances due to the GRF are then approximated by a sequence of jointly Gaussian random vectors, the statistics of which depend on both the structure of the GRF and the nominal trajectory.
%We thus turn to the theory of controlling linear stochastic systems with Gaussian disturbances.
Thus, the linearized optimal control subproblem is reduced to the more tractable situation of a linear system being affected by temporal disturbances.

For linear stochastic systems with additive Gaussian disturbances, and in absence of any state or control constraints, it is well known that the nominal control steers the state mean while the feedback gains steer the state covariance \cite{Okamoto2018csl}.
State or control constraints have to be imposed as probabilistic (e.g, chance) constraints since the system is stochastic. 
Chance constraints, however, depend on both the state mean and the covariance.
Thus, the chance-constrained optimal control of a linear stochastic system involves a joint optimization over the nominal control and the feedback gains.
This problem is referred to as chance-constrained covariance steering, since the control law is designed to explicitly steer the dynamics of the state covariance \cite{Chen2016a,Ridderhof2021minfuel}.
Previous works have shown that state history feedback laws result in a convex formulation of the chance constrained covariance steering problem \cite{Bakolas2016cs,Okamoto2018csl,Ridderhof2019nonlinear,Ridderhof2020ofcs}.
For the present problem, we may therefore jointly optimize updates to the nominal control and the feedback gains, while considering the local effect of the GRF-induced disturbances, and while enforcing the problem chance constraints.
Finally, the optimal control from each linearized subproblem is used to propagate the nominal, nonlinear dynamics to obtain the reference trajectory for the subsequent iterate.  

%Has been applied for spacecraft controls \cite{Ridderhof2021minfuel,Ridderhof2020lowthrust}

The contributions of this paper include: a) the derivation of discrete-time Gaussian disturbances resulting from motion through a spatially-defined GRF; 
b) the development of a successive convex programming approach to solve the resulting chance-constrained stochastic optimal control problem; and c) the application of the developed theory to the problem of aerocapture guidance.
Specifically, the existing chance-constrained stochastic optimal control literature treats discrete Gaussian disturbances as having fixed and specified statistics; in this paper, we show that the statistics of discrete-time Gaussian disturbances can be derived from the motion of a system through a GRF.
For the aerocapture problem, in particular, the approach results in both a novel analytical quantification of vehicle trajectory covariance due to density variations and a table-lookup-based guidance scheme that includes closed-loop probabilistic constraints.

This paper is organized as follows.
Properties of GRFs are briefly reviewed in Section~\ref{sec:grf}.
The stochastic optimal control problem of chance-constrained covariance steering in a GRF is introduced in Section~\ref{sec:problem_formulation}.
In Section~\ref{sec:solution_method}, a solution to this problem is developed by successive convexification.
The proposed method is first demonstrated on a simple double integrator problem in Section~\ref{sec:example_double_integrator}, and it is then applied to the aerocapure guidance problem in Section~\ref{sec:example_areocap}.
Finally, Section~\ref{sec:conclusion} summarizes the results of the paper and suggests some potential extensions and research directions.

\section{Gaussian Random Fields}
\label{sec:grf}

A collection of random variables $\{ \Psi(z) : z \in \Iset \}$ is a \textit{Gaussian random field} (GRF), also referred to as a Gaussian process, if any finite linear combination of the variables $\Psi(z_i)$ with $\{z_i\} \subset \Iset$ is Gaussian distributed  --- that is, if the variables $\Psi(z_i)$ are \textit{jointly Gaussian} \cite{LeGall2016,Rasmussen2006gp}.
In other words, each element $z$ in an index set $\Iset$ (for example, $\Iset = \Re^d$) determines a Gaussian random variable $\Psi(z)$, and, in addition, any finite collection random variables $\Psi(z_i)$ determined by the inputs $z_i$ are jointly Gaussian.
Henceforth, we will often refer to the collection $\{ \Psi(z) : z \in \Iset \}$ simply as $\Psi$ when the context is clear.
A GRF is fully characterized by a mean function
\begin{equation}
	\mu: \Iset \to \Re, \quad \mu(z) = \E\big( \Psi(z) \big),
\end{equation}
and a positive semi-definite covariance function 
\begin{equation}
	\covfun : \Iset \times \Iset \to \Re, \quad \covfun(z_1, z_2) = \Cov \big( \Psi(z_1), \Psi(z_2) \big).
\end{equation}
Thus, the values of the field $\Psi_* = \big( \Psi(z_{1}^*), \dots, \Psi(z_{n}^*) \big)$ at any $n$ input points $\{z^*_1, \dots, z^*_n\} \subset \Iset$ are Gaussian distributed as $\Psi_* \sim \Ncal (\mu_*, \Sigma_{*,*})$, where
\begin{equation}
	\mu_* = \begin{bmatrix}
		\mu(z^*_1) \\ \vdots \\ \mu(z^*_n)
	\end{bmatrix}, \qquad
	\Sigma_{*,*} = \begin{bmatrix}
		\covfun(z^*_1, z^*_1) & \cdots & \covfun(z^*_1, z^*_n) \\
		\vdots & \ddots & \vdots \\
		\covfun(z^*_n, z^*_1) & \cdots & \covfun(z^*_n, z^*_n) \\
	\end{bmatrix}.
\end{equation}
%A useful property of GRFs is the ability to efficiently compute conditional distributions based on noisy measurements.

GRFs are often used in the context of conditioning based on noisy measurements.
For example, samples from one and two-dimensional GRFs, with and without conditioning on measurements, are shown in Figures~\ref{fig:gp_cond_example} and \ref{fig:field_cond}.
In this paper, however, only the aforementioned jointly Gaussian property of GRF samples will be used.
The interested reader is referred to Ref.~\cite{Rasmussen2006gp} for more details on GRFs.

\begin{figure}
	\centering
	% trim={<left> <lower> <right> <upper>}
	\includegraphics[trim={.05in 0.0in 0.05in 0.0in}, clip]{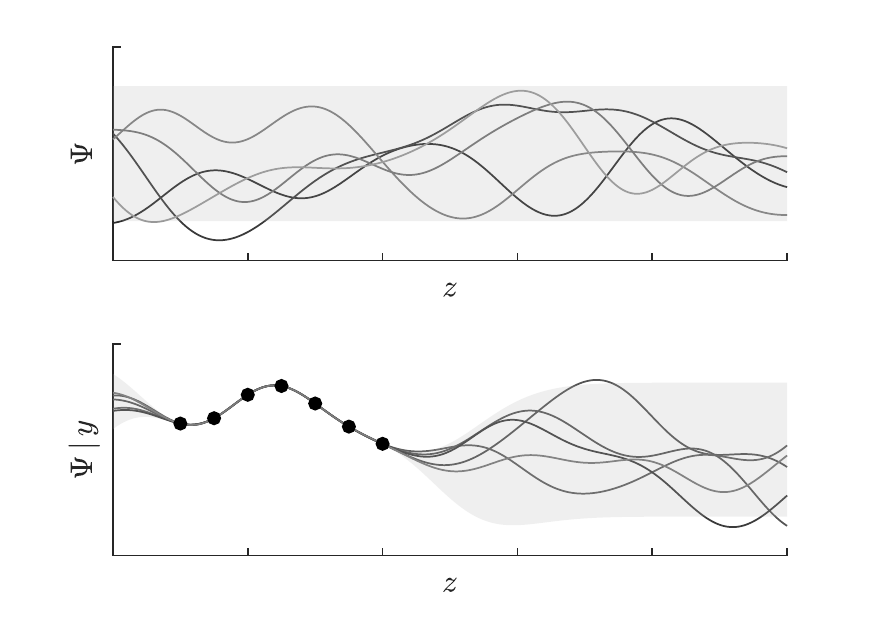}
	\caption{
		Single-dimensional GRF $\Psi$ conditioned on measurements
		\label{fig:gp_cond_example}}
\end{figure}

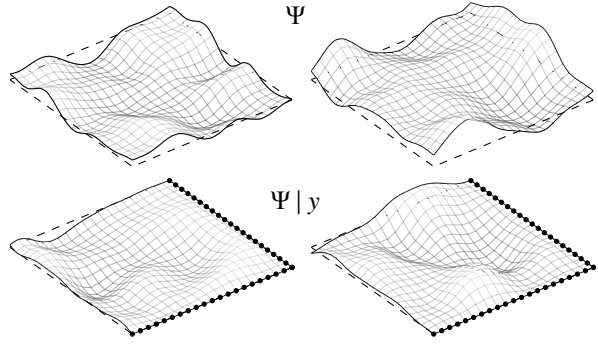
\begin{figure}
	\centering
	\tikzsetnextfilename{field_cond}
	\begin{tikzpicture}[]
    
\begin{groupplot}[
	group style={group size=2 by 2, horizontal sep=0.1in, vertical sep=-0.5in},
	height=2in,
	width=2.1in,
%	scale only axis,
%	small,
	hide axis,
	colormap/blackwhite,
	zmin = -0.5,
	zmax = 0.5,
	view={-37.5}{45},
	]
	\pgfplotsset{compat=1.8}
	\nextgroupplot
	\addplot3[dashed] table {
		1 0 0
		1 1 0
		0 1 0
	};
	\addplot3[
		surf,
		opacity=0.2,
		fill opacity=1.0,
		color = white,
		mesh/rows=21,mesh/ordering=y varies,
	] table {tikzsrc/data/p1_1_P.txt};
	\addplot3[smooth] table [header=false] {tikzsrc/data/p1_1_l1.txt};
	\addplot3[smooth] table [header=false] {tikzsrc/data/p1_1_l2.txt};
	\addplot3[smooth] table [header=false] {tikzsrc/data/p1_1_l3.txt};
	\addplot3[smooth] table [header=false] {tikzsrc/data/p1_1_l4.txt};
	\addplot3[dashed] table {
			1 0 0
			0 0 0
			0 1 0
	};

	\nextgroupplot
	\addplot3[dashed] table {
		1 0 0
		1 1 0
		0 1 0
	};
	\addplot3[smooth] table [header=false] {tikzsrc/data/p1_2_l1.txt};
	\addplot3[smooth] table [header=false] {tikzsrc/data/p1_2_l2.txt};
	\addplot3[smooth] table [header=false] {tikzsrc/data/p1_2_l3.txt};
	\addplot3[smooth] table [header=false] {tikzsrc/data/p1_2_l4.txt};
	\addplot3[
		surf,
		opacity=0.2,
		fill opacity=1.0,
		color = white,
		mesh/rows=21,mesh/ordering=y varies,
	] table {tikzsrc/data/p1_2_P.txt};
	\addplot3[dashed] table {
			1 0 0
			0 0 0
			0 1 0
	};
	
	\nextgroupplot
	\addplot3[dashed] table {
		1 0 0
		1 1 0
		0 1 0
	};
	\addplot3[smooth] table [header=false] {tikzsrc/data/p2_1_l1.txt};
	\addplot3[smooth] table [header=false] {tikzsrc/data/p2_1_l2.txt};
	\addplot3[smooth] table [header=false] {tikzsrc/data/p2_1_l3.txt};
	\addplot3[smooth] table [header=false] {tikzsrc/data/p2_1_l4.txt};
	\addplot3[
		surf,
		opacity=0.2,
		fill opacity=1.0,
		color = white,
		mesh/rows=21,mesh/ordering=y varies,
	] table {tikzsrc/data/p2_1_P.txt};
	\addplot3[dashed] table {
			1 0 0
			0 0 0
			0 1 0
	};
	\addplot3 [only marks, mark=*, mark color = black, fill=black, mark size = 0.8pt] table [header=false] {tikzsrc/data/p2_1_l1.txt};
	\addplot3 [only marks, mark=*, mark color = black, fill=black, mark size = 0.8pt] table [header=false] {tikzsrc/data/p2_1_l4.txt};
	
	\nextgroupplot
	\addplot3[dashed] table {
		1 0 0
		1 1 0
		0 1 0
	};
	\addplot3[smooth] table [header=false] {tikzsrc/data/p2_2_l1.txt};
	\addplot3[smooth] table [header=false] {tikzsrc/data/p2_2_l2.txt};
	\addplot3[smooth] table [header=false] {tikzsrc/data/p2_2_l3.txt};
	\addplot3[smooth] table [header=false] {tikzsrc/data/p2_2_l4.txt};
	\addplot3[
		surf,
		opacity=0.2,
		fill opacity=1.0,
		color = white,
		mesh/rows=21,mesh/ordering=y varies,
	] table {tikzsrc/data/p2_2_P.txt};
	\addplot3[dashed] table {
			1 0 0
			0 0 0
			0 1 0
	};
	\addplot3 [only marks, mark=*, mark color = black, fill=black, mark size = 0.8pt] table [header=false] {tikzsrc/data/p2_2_l1.txt};
	\addplot3 [only marks, mark=*, mark color = black, fill=black, mark size = 0.8pt] table [header=false] {tikzsrc/data/p2_2_l4.txt};
\end{groupplot}

%\draw[step=1.0,black,thin] (0, -2) grid (7,4);
\node[] at (3.8, 2.75) {$\Psi$};
\node[] at (3.8, 0.25) {$\Psi \, \vert \, y$};

\end{tikzpicture}
	\caption{
		Two-dimensional GRF $\Psi$ conditioned to have the right-most edges be constant
		\label{fig:field_cond}}
\end{figure}

\section{Problem Formulation}
\label{sec:problem_formulation}

Consider a system with state $x \in \Re ^ n$, and let
%$\Psi : \Re ^ d \to \Re$
$
	\{ \Psi(z) \in \Re : z \in \Re ^ d \}
$
be a GRF with known mean function
$\mu : \Re ^ d \to \Re$ and known covariance function $\covfun : \Re ^ d \times \Re ^ d \to \Re$.
%$\mu$ and known covariance function $Q$.
The independent variable $z$ of the GRF $\Psi$ is a function of the system state, given by $z = \phi(x)$.
Let the system state evolve according to
\begin{equation} \label{eq:system}
	\dot{x} = f\big(x, u, \Psi(\phi(x)) \big),
\end{equation}
with the initial condition
\begin{equation}
	x(t_0) \sim \Ncal ( \bar{x}_0, P_0 ),
\end{equation}
where $u \in \Re ^ m$ is the control input, and where the mean vector $\bar{x}_0$ and covariance matrix $P_0$ are both fixed and known.
The initial state $x_0$ is assumed to be independent of the field $\Psi$.
The evolution of the system (\ref{eq:system}) is considered on the discrete time partition
\begin{equation}
	\mathscr{P} = ( t_0, \dots, t_N ),
\end{equation}
where $t_0 < t_1 < \cdots < t_N = t_f$
for a given time horizon $N$, and such that $t_f > 0$ is a given, fixed final time.

The control is assumed to be piecewise constant on subintervals of the partition $\mathscr{P}$, so that
\begin{equation}
	u(t) = u(t_k), \quad  \forall t \in [t_k, t_{k + 1} ), \quad k = 0, \dots, N - 1.
\end{equation}
We write $x_k = x(t_k)$ and $u_k = u(t_k)$ for notational simplicity.
The control is assumed to follow the state history feedback law
\begin{equation} \label{eq:control_law_orig}
	u_k = \sum_{\ell = 0}^{k} K_{k, \ell} \tilde{x}_\ell + v_k,
\end{equation}
where $\tilde{x}_\ell = x_\ell - \bar{x}_\ell$ is the state deviation from its mean, $K_{k, \ell} \in \Re ^ {m \times n}$ are feedback gains, and where $v_k \in \Re ^ m$ are nominal controls.
As will be shown in the following sections, state history feedback results in a convex formulation of the chance-constrained covariance steering problem.
Intuitively, and in contrast to Brownian-disturbance driven processes, state history feedback is required since, due to the GRF $\Psi$, the state process may not be Markovian.
For example, if the system $x(t_2)$ returns to a previously visited state $x(t_1)$, for some $t_2 > t_1$, then the value of the state $x(t_1)$ may add information about the disturbance experienced at time $t_1$ --- a violation of the Markov assumption.

%
% \panos{any justification of this statement? This is not clear and I am not sure whether it is correct}

The state and controls are required to satisfy the chance constraints
\begin{equation} \label{eq:state_chance_constraint}
	\Pr \big( a_{i, k} \t x_k \geq \alpha_{i, k} \big) \leq p^x_{i, k}, \quad \forall (i, k) \in \mathcal{X},
\end{equation}
\begin{equation} \label{eq:control_chance_constraint}
	\Pr \big( b_{i, k} \t u_k \geq \beta_{i, k} \big) \leq p^u_{i, k}, \quad \forall (i, k) \in \mathcal{U},
\end{equation}
where the vectors $a_{i, k} \in \Re ^ n$, $b_{i, k} \in \Re ^ m$ and scalars $\alpha_{i, k}$, $\beta_{i, k}$ define half-plane constraints, and $p^x_{i, k}, p^u_{i, k} \in (0, 0.5)$ are maximum probabilities of constraint violation.
The index sets $\mathcal{X}$ and $\mathcal{U}$ determine the number of half-plane constraints to enforce at each decision time $t_k$.
Furthermore, the mean and covariance of the state at the final time $x_f = x(t_f)$ are constrained by
\begin{subequations} \label{eq:final_state_distribution_constraints}
	\begin{equation} \label{eq:final_mean_constraint}
		\E(x_f) = \bar{x}_f,
	\end{equation}
	\begin{equation}
		\Cov(x_f) \leq P_f,
	\end{equation}
\end{subequations}
for a given target mean state $\bar{x}_f$ and positive definite maximum final covariance matrix $P_f$.
Subject to the aforementioned constraints, we are concerned with finding the feedback gains $K_{k, \ell}$ and the feedforward controls $v_k$ to minimize the quadratic cost 
\begin{equation} \label{eq:J_1_def}
	J_1(K_{k, \ell}, v_k) = \E \bigg( \sum_{k = 0}^{N - 1} (x_k - x^d_k) \t Q_k (x_k - x^d_k) + \tilde{u}_k \t R_k \tilde{u}_k \bigg) +  \sum_{k = 0}^{N - 1} \bar{u}_k \t \bar{R}_k \bar{u}_k,
\end{equation}
for user-defined state and control weight matrices $Q_k \geq 0$ and $R_k, \bar{R}_k \geq 0$, and where $x_k^d$ is a given defined desired trajectory.
The cost weight is separated into $R_k$ and $\bar{R}_k$ so that, if desired, the control variance may be penalized without penalizing the nominal control.
Alternatively, the upper $1 - p_f$ percentile of a functional of the final state may be minimized by considering the cost
\begin{equation} \label{eq:J_2_def}
	J_2(K_{k, \ell}, v_k) = \inf \{ \gamma \in \Re : \Pr (\xi \t x_f > \gamma) \leq p_f \},
\end{equation}
where $\xi \in \Re ^ n$ and $p_f \in (0, 1)$ are user-defined constants.
%More generally, any cost function which is quadratic in the state or control variations, and which is convex in the state and control means, may be considered.
Note that, when seeking to minimize the upper percentile cost (\ref{eq:J_2_def}), the final state mean should not be constrained, since changing the final state mean may affect the cost value.

Without loss of generality, the cost is taken to be the weighted sum
% we take the cost to be the weighted sum
\begin{equation}
	J = J_1 + \eta J_2,
\end{equation}
for some non-negative scalar $\eta$.
Indeed, setting $Q_k$, $R_k$, and $\bar{R}_k$ to zero and $\eta = 1$ recovers the $1 - p_f$ percentile cost (\ref{eq:J_2_def}), whereas setting $\eta = 0$ results in the purely quadratic cost (\ref{eq:J_1_def}).

%\subsection{Conditioning on Measurements}

%In the following 

\section{Solution via Successive Convex Programming}
\label{sec:solution_method}

\subsection{Approximation About a Nominal Trajectory}

Assume that a nominal control input $\hat{u}$ is provided on the time interval $[t_0, t_f]$, and let the corresponding nominal state be the solution to the system
\begin{equation} \label{eq:nominal_system}
	\dot{\hat{x}} = f\big(\hat{x}, \hat{u}, \E(\Psi(\phi(\hat{x}))) \big), 
\end{equation}
with the initial value $\hat{x}(t_0) = \bar{x}_0$.
The GRF $\Psi$, its mean function $\mu$, and its covariance function $\covfun$, evaluated along the nominal trajectory $\hat{x}$, are denoted as
\begin{subequations} \label{eq:GRF_hat}
	\begin{equation}
		\hat{\Psi}(t) = \Psi\big( \phi(\hat{x}(t)) \big),
	\end{equation}
	\begin{equation}
		\hat{\mu}(t) = \mu \big( \phi( \hat{x}(t) ) \big),
	\end{equation}
	\begin{equation}
		\hat{\covfun}(t, \tau) = \covfun \big( \phi( \hat{x}(t) ), \phi( \hat{x}(\tau) ) \big).
	\end{equation}
\end{subequations}
Note that, unlike $\hat{x}$ and $\hat{u}$, the function $\hat{\Psi}$ is random: the function $\hat{\Psi}$ is an approximation of $\Psi$ in the sense that the statistics of $\hat{\Psi}$ are evaluated along the nominal trajectory rather than the perturbed trajectory.
In other words, the nominal trajectory $\hat{x}$ determines a mapping from the spatially-dependent random field $\Psi$ to the time-dependent random process $\hat{\Psi}$; this relationship is shown graphically in Figure~\ref{fig:field_diagram}.
The following result establishes the consistency of the \mbox{definitions (\ref{eq:GRF_hat}).}

\begin{figure}
	\centering
	\tikzsetnextfilename{field_diagram}
	\begin{tikzpicture}[
	declare function={f(\x,\y)=0.6 + 0.025 * (\y * sin(deg(\x)) - \x * cos(deg(\y)));}
]

\usetikzlibrary{decorations.markings}
\tikzset{
  set arrow inside/.code={\pgfqkeys{/tikz/arrow inside}{#1}},
  set arrow inside={end/.initial=>, opt/.initial=},
  /pgf/decoration/Mark/.style={
    mark/.expanded=at position #1 with
    {
      \noexpand\arrow[\pgfkeysvalueof{/tikz/arrow inside/opt}]{\pgfkeysvalueof{/tikz/arrow inside/end}}
    }
  },
  arrow inside/.style 2 args={
    set arrow inside={#1},
    postaction={
      decorate,decoration={
        markings,Mark/.list={#2}
      }
    }
  },
}

\begin{axis}[
	small,
	hide axis,
	colormap/blackwhite,
	zmin = 0,
	zmax = 0.75,
]
\addplot3[
	surf,
	domain=-3:3,
	domain y=-3:3,
	opacity=0.1,
	fill opacity=0.0,
	samples=20,
]
	{f(x,y)};

% Space curve and field curve
\addplot3 [arrow inside={}{0.2, 0.4, ..., 0.8}] table [x=X, y=Y, z expr={ f(\thisrow{X}, \thisrow{Y}) }] {tikzsrc/line_data.txt};
\addplot3 [arrow inside={}{0.2, 0.4, ..., 0.8}] table [x=X, y=Y, z expr={0}] {tikzsrc/line_data.txt};

% Edges of field
\addplot3[domain=-3:3, samples=60,samples y=0] ((x, 3, {f(x,3)});
\addplot3[domain=-3:3, samples=60,samples y=0] ((x, -3, {f(x,-3)});
\addplot3[y domain=-3:3, samples=0,samples y=60] ((-3, y, {f(-3,y)});
\addplot3[y domain=-3:3, samples=0,samples y=60] ((3, y, {f(3,y)});

%\node[coordinate,pin=above:{$\Psi$}] at (axis cs:-3, 3, 0.5) {};
%\node[] at (axis description cs:0.1,0.8) {$\Psi(z)$};
%\node[] at (axis description cs:0.75,0.78) {$\hat{\Psi}(t)$};
%\node[] at (axis description cs:0.18,0.25) {$\hat{x}(t)$};
%\node[] at (axis description cs:0.5,0.25) {$\phi$};

% Guidelines
%\addplot3[
%]
%coordinates{
%(-3,-3,0) (-3,3,0) (3,3,0) (3,-3,0) (-3,-3,0)
%};
%\addplot3[
%]
%coordinates{
%(-3,-3,0) (-3,-3,0.5)
%};
\end{axis}
%\draw[step=1.0,black,thin] (0, 0) grid (5,5);

\node[] at (0.75, 3.5) {$\Psi(z)$};
\node[] at (3.8,2.6) {$\hat{\Psi}(t)$};
\node[] at (3.25, 0.6) {$\hat{x}(t)$};
%\node[] at (axis description cs:0.5,0.25) {$\phi$};

\draw [->] (1.2, 1) to[bend left=+20] node[right]{$\phi$} (1.5, 2);

\end{tikzpicture}
	\caption{
		Samples of the GRF $\Psi$ and random process $\hat{\Psi}$ along the nominal trajectory $\hat{x}$
		\label{fig:field_diagram}}
\end{figure}
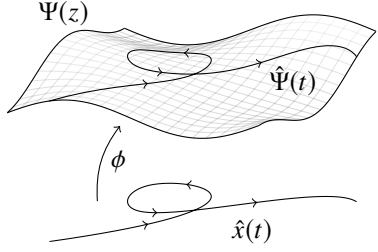

\begin{prop}
	The function $\hat{\Psi}(t)$ is a Gaussian random process with mean $\hat{\mu}(t)$ and covariance $\hat{\covfun}(t, \tau)$.
\end{prop}

\begin{proof}
	The process $\hat{\Psi}$ is Gaussian since the random field $\Psi$ is Gaussian; it remains only to show that the mean and covariance of $\hat{\Psi}$ are given by $\hat{\mu}$ and $\hat{\covfun}$.
	From the definitions (\ref{eq:GRF_hat}), we obtain
	\begin{equation} \label{eq:Psihat_mean}
		\E \big( \hat{\Psi}(t) \big) = \E \big( \Psi \big( \phi( \hat{x}(t) ) \big) \big) = \mu \big( \phi( \hat{x}(t) ) \big) = \hat{\mu}(t),
	\end{equation}
	and
	\begin{align} \label{eq:Psihat_cov}
		\Cov \big( \hat{\Psi}(t), \hat{\Psi}(\tau) \big) &= \Cov \big( \Psi \big( \phi( \hat{x}(t) ) \big), \Psi \big( \phi( \hat{x}(\tau) ) \big) \big) \nonumber \\
		&= \covfun \big( \phi( \hat{x}(t) ), \phi( \hat{x}(\tau) ) \big) \nonumber \\
		&= \hat{\covfun}(t, \tau),
	\end{align}
	which yields the desired result.
\end{proof}

Next, we linearly approximate the system dynamics about the nominal trajectory and mean disturbance $\big(\hat{x}, \hat{u}, \hat{\mu} \big)$ to obtain
\begin{equation} \label{eq:linearized_system}
	\dot{x} \approx f(\hat{x}, \hat{u}, \hat{\mu}) + \frac{\partial f}{\partial x} (x - \hat{x}) + \frac{\partial f}{\partial u} (u - \hat{u}) + \frac{\partial f}{\partial \Psi} (\hat{\Psi} - \hat{\mu}).
\end{equation}
Define the functions
\begin{equation}
	A(t) = \frac{\partial f}{\partial x}, \quad B(t) = \frac{\partial f}{\partial u}, \quad G(t) = \frac{\partial f}{\partial \Psi},
\end{equation}
evaluated at $\big( \hat{x}(t), \hat{u}(t), \hat{\mu}(t) \big)$, and let
\begin{equation}
	c(t) = f(\hat{x}, \hat{u}, \hat{\mu}) - A(t) \hat{x} - B(t) \hat{u} - G(t) \hat{\mu}.
\end{equation}
The linearized system (\ref{eq:linearized_system}) is integrated from time $t_k$ to $t_{k + 1}$ to obtain the approximate system evolution
\begin{equation} \label{eq:discrete_system_integral}
	x_{k + 1} \approx \Phi(t_{k + 1}, t_k) x_k
	+ \int_{t_k}^{t_{k + 1}} \Phi(t_{k + 1}, t) \big( B(t) u_k + c(t) \big) \, \d t
	+ \int_{t_k}^{t_{k + 1}} \Phi(t_{k + 1}, t) G(t) \hat{\Psi}(t) \, \d t,
\end{equation}
where $\Phi$ is the state transition matrix corresponding to $A(t)$.
Simplifying, (\ref{eq:discrete_system_integral}) is written as the stochastic difference equation
\begin{equation} \label{eq:discrete_system}
	x_{k + 1} = A_k x_k + B_k u_k + c_k + w_k,
\end{equation}
with the values $A_k$, $B_k$ and $c_k$ taken from (\ref{eq:discrete_system_integral}), and where the Gaussian disturbance term $w_k$ is given by
\begin{equation} \label{eq:w_k_def}
	w_k = \int_{t_k}^{t_{k + 1}} \Phi(t_{k + 1}, t) G(t) \hat{\Psi}(t)  \, \d t.
\end{equation}
As shown in the following result, the mean and covariance of the disturbance term $w_k$ in (\ref{eq:discrete_system}) depend on the system dynamics and on the statistics of the GRF $\Psi$, via the functions $\hat{\mu}$ and $\hat{\covfun}$.

\begin{prop}
	The vectors $w_k$, for $k = 0, \dots, N - 1$, are jointly Gaussian with mean values
	\begin{equation} \label{eq:disturbance_mean}
		\E( w_k ) = \int_{t_k}^{t_{k + 1}}  \Phi(t_{k + 1}, t) G(t) \hat{\mu}(t) \, \d t,
	\end{equation}
	and covariances
	\begin{equation} \label{eq:disturbance_covariance}
		\Cov (w_k, w_\ell) =
		\int_{t_k}^{t_{k + 1}} \int_{t_\ell}^{t_{\ell + 1}} \Phi(t_{k + 1}, t) G(t) 
		\hat{\covfun}(t, \tau) G \t (\tau) \Phi \t (t_{\ell + 1}, \tau) \, \d \tau \, \d t.
	\end{equation}
\end{prop}

\begin{proof}	
	The mean term (\ref{eq:disturbance_mean}) follows from taking the expectation of $w_k$ in (\ref{eq:w_k_def}) and substituting (\ref{eq:Psihat_mean}).
	Furthermore,
	\begin{equation} \label{eq:w_k_error}
		w_k - \E(w_k) = \int_{t_k}^{t_{k + 1}}  \Phi(t_{k + 1}, t) G(t) \big( \hat{\Psi}(t) - \hat{\mu}(t) \big) \, \d t.
	\end{equation}
	The covariance of $w_k$ and $w_\ell$ is computed from (\ref{eq:w_k_error}) as
	\begin{align}
		\Cov(w_k, w_\ell) &= \E \big\{ \big( w_k - \E(w_k)  \big) \big(w_\ell - \E(w_\ell)  \big) \t \big\} \nonumber \\
		&= \E \bigg\{ \int_{t_k}^{t_{k + 1}}  \Phi(t_{k + 1}, t) G(t) \big( \hat{\Psi}(t) - \hat{\mu}(t) \big) \, \d t \int_{t_\ell}^{t_{\ell + 1}} \big( \hat{\Psi}(\tau) - \hat{\mu}(\tau) \big) G \t (\tau) \Phi \t (t_{\ell + 1}, \tau)  \, \d \tau \bigg\} \nonumber \\
		&= \int_{t_k}^{t_{k + 1}} \int_{t_\ell}^{t_{\ell + 1}} \Phi(t_{k + 1}, t) G(t)
		\E \big\{ \big( \hat{\Psi}(t) - \hat{\mu}(t) \big) \big( \hat{\Psi}(\tau) - \hat{\mu}(\tau) \big) \big\}
		G \t (\tau) \Phi \t (t_{\ell + 1}, \tau)  \, \d \tau \, \d t \label{eq:w_cov_final_line}.
	\end{align}
	Substituting the covariance function $\hat{\covfun}$ from (\ref{eq:Psihat_cov}) into (\ref{eq:w_cov_final_line}), we obtain the desired result.
	Finally, by the definition of a GRF, any finite collection of evaluations of $\big( \Psi(z_i) \big)$ are jointly Gaussian, and thus integrals over $\Psi$ are also jointly Gaussian.
\end{proof}

\begin{remark}
    The system (\ref{eq:discrete_system}) is simply a stochastic difference equation with Gaussian disturbances, but the disturbances are neither (necessarily) independent nor identically distributed.
    In contrast to the problem often treated in the stochastic control literature, the statistics of the Gaussian disturbances in (\ref{eq:discrete_system}) are derived by both the nominal system motion and by the GRF statistics.
\end{remark}

% \jack{make into remark and add text to emphasize this point}
% \panos{This is an important point that needs to be stressed further as distrinquishes this work with more standard stochastic control problems}
%Fortunately, covariance steering theory \cite{Okamoto2018csl,Ridderhof2019nonlinear,Ridderhof2020ofcs}, which is relies on state history feedback, permits any jointly-Gaussian disturbance sequence.
%In the following section, we build from existing results to control the system (\ref{eq:discrete_system}).

\subsection{Block-Matrix Formulation}

The state process (\ref{eq:discrete_system}) may be equivalently written in block-matrix notation as \cite{Skaf2010design,Okamoto2018csl,Bakolas2016cs}
\begin{equation} \label{eq:state_matrix_form}
	\begin{bmatrix}
		x_0 \\ x_1 \\ x_2 \\ \vdots
	\end{bmatrix}
	= \begin{bmatrix}
		I \\ A_0 \\ A_1 A_0 \\ \vdots
	\end{bmatrix} x_0
	+ \begin{bmatrix}
		0 & 0 &  \\
		B_0 & 0 &  \\
		A_1 B_0 & B_1 &  \\
		&& \ddots
	\end{bmatrix} \begin{bmatrix}
		u_0 \\ u_1 \\ \vdots
	\end{bmatrix}
	+ \begin{bmatrix}
		0 \\ c_0 \\ A_1 c_1 \\ \vdots
	\end{bmatrix}
	+ \begin{bmatrix}
		0 & 0 &  \\
		I & 0 &  \\
		A_1 & I &  \\
		&& \ddots
	\end{bmatrix} \begin{bmatrix}
		w_0 \\ w_1 \\ \vdots
	\end{bmatrix}.
\end{equation}
Let $X$ be a column vector constructed by stacking the states $x_k$ for $k = 0, 1, \dots, N$, and, similarly, let $U$ and $W$ be the column vectors constructed by stacking the controls $u_k$ and disturbances $w_k$ for $k = 0, 1, \dots, N - 1$.
For appropriately constructed block matrices $\Amat$, $\Bmat$, and $\Gmat$ as in (\ref{eq:state_matrix_form}), and with $C$ an appropriately constructed vector, the state process can be written as the linear matrix equation
\begin{equation} \label{eq:discrete_block_system}
	X= \Amat x_0 + \Bmat U + C + \Gmat W.
\end{equation}
See Refs.~\cite{Okamoto2018csl,Skaf2010design,Bakolas2016cs} for details on this construction.
%Let $V$ be column a vector defined by stacking the feedforwrad inputs $(v_k)$ as $U$, and let
Letting the block lower-triangular matrix $\Kmat \in \Re^{Nm\times(N+1)n}$ be given by
\begin{equation} \label{eq:L_matrix}
	\Kmat = \begin{bmatrix}
		K_{0,0} & 0 & & \cdots & 0 \\
		K_{1, 0} & K_{1, 1} & 0 & \cdots & 0 \\
		\vdots & & & & \vdots \\
		K_{N - 1, 0} & K_{N - 1, 1} & K_{N - 1, 2} & \cdots & 0
	\end{bmatrix},
\end{equation}
and letting $U \in \Re ^{Nm}$, $V \in \Re^{N m}$, and $\tilde{X} \in \Re ^{(N+1)n}$ be the vectors obtained by stacking the closed-loop controls $(u_k)$, the feedforward controls $(v_k)$, and the state deviation $(\tilde{x}_k)$, the control law (\ref{eq:control_law_orig}) is given in block-matrix notation as
\begin{equation} \label{eq:state_feedback_block}
	U = \Kmat \tilde{X} + V.
\end{equation}
Substituting the control (\ref{eq:state_feedback_block}) into the state equation (\ref{eq:discrete_block_system}) gives the closed-loop system
\begin{equation}
	\bar{X} = \Amat \bar{x}_0 + \Bmat V + C + \Gmat \bar{W},
\end{equation}
\begin{equation}
	\tilde{X} = (I - \Bmat \Kmat) \inv (\Amat \tilde{x}_0 + \Gmat \tilde{W}).
\end{equation}
Note that the mean state $\bar{X}$ depends only on the nominal control $V$, whereas the random state deviation $\tilde{X}$ depends only on the feedback gain $\Kmat$.
%\begin{equation}
%	\Lmat = \Kmat (I - \Bmat \Kmat) \inv
%\end{equation}

Following~\cite{Skaf2010design}, we define the new decision variable $\Lmat \in \Re ^{N m \times (N + 1) n}$ as
%\begin{equation} \label{ch2:eq:Fmat_def}
%	\Fmat = \Kmat (I - \Bmat \Kmat) \inv \in \Re ^{N n_u \times (N + 1) n_x}.
%\end{equation}
\begin{equation} \label{eq:L_def}
	\Lmat = \Kmat (I - \Bmat \Kmat) \inv.
\end{equation}
Since $\Kmat$ is block lower-triangular and $\Bmat$ is strictly block lower-triangular, the matrix $I - \Bmat \Kmat$ is invertible.
It follows that $\Lmat$ is block lower-triangular and satisfies
\begin{equation} \label{ch2:eq:K_F_equality}
	I + \Bmat \Lmat = (I - \Bmat \Kmat) \inv,
\end{equation}
%Furthermore, $\Kmat$ is a function of $\Fmat$ given by
\begin{equation}
	\Kmat = \Lmat (I + \Bmat \Lmat) \inv.
\end{equation}
Therefore, we optimize over $\Lmat$ in place of $\Kmat$ \cite{Skaf2010design}.

Using the decision variable $\Lmat$ as in (\ref{eq:L_def}) results in the closed-loop system
\begin{equation} \label{eq:state_mean}
	\bar{X} = \Amat \bar{x}_0 + \Bmat V + C + \Gmat \bar{W},
\end{equation}
\begin{equation} \label{eq:state_deviation}
	\tilde{X} = (I + \Bmat \Lmat) (\Amat \tilde{x}_0 + \Gmat \tilde{W}).
\end{equation}
The state and control processes $X$ and $U$ are thus approximately, due to the linerization, Gaussian distributed with mean $\E(X) = \bar{X}$ as in (\ref{eq:state_mean}), $\E(U) = V$, and covariances
\begin{equation} \label{eq:state_covariance}
	\PXmat = \Cov(X) = (I + \Bmat \Lmat) \Smat (I + \Bmat \Lmat) \t,
\end{equation}
\begin{equation}
	\PUmat = \Cov(U) = \Lmat \Smat \Lmat \t,
\end{equation}
where
\begin{equation} \label{eq:Smat_def}
	\Smat = \Amat P_0 \Amat \t + \Gmat \Cov(W) \Gmat \t.
\end{equation}
The elements of the mean disturbance vector $\bar{W}$ and the covariance matrix $\Cov(W)$ are obtained from the integrals (\ref{eq:disturbance_mean}) and (\ref{eq:disturbance_covariance}).

\subsection{Chance Constraints} \label{sec:cvx_cc}

Consider next the state chance constraint (\ref{eq:state_chance_constraint}).
Notice that
the inner product $a_{i, k} \t x_k$ is a Gaussian random variable with mean $a_{i, k} \t \E(x_k)$ and covariance $a_{i, k} \t \Cov(x_k) a_{i, k}$.
%It follows that the compliment of the probability in (\ref{ch2:eq:relaxed_state_chance_constraint}) can be written in terms of the normal cumulative distribution function $\ncdf$ as
It follows that
\begin{equation}
	\Pr (a_{i, k} \t x_k \leq \alpha_{i, k}) = \ncdf \bigg( \frac{\alpha_{i, k} - a_{i, k} \t E_k \bar{X}}{\sqrt{a_{i, k} \t E_k \PXmat E_k \t a_{i, k}}} \bigg),
\end{equation}
where $\ncdf$ is the normal cumulative distribution function.
Taking the inverse of the normal cumulative distribution function and rearranging terms, we obtain
%\begin{align}
%	\Pr (a_j \t x_k > \alpha_j) \leq p^x_j &\iff
%	\ncdf \bigg( \frac{\alpha_j - a_j \t E_k \bar{X}}{\sqrt{a_j \t E_k \PXmat E_k \t a_j}} \bigg) \geq 1 - p^x_j \\
%	&\iff \frac{\alpha_j- a_j \t E_k \bar{X}}{\sqrt{a_j \t E_k \PXmat E_k \t a_j}} \geq \ncdf \inv ( 1 - p^x_j) \\
%	&\iff \ncdf \inv ( 1 - p^x_j) \sqrt{a_j \t E_k \PXmat E_k \t a_j} + a_j \t E_k \bar{X} - \alpha_j \leq 0. \label{ch2:eq:control_chance_constraint_with_sqrt}
%\end{align}
%\begin{multline}
%	\Pr (a_{i, k} \t x_k > \alpha_{i, k}) \leq p^x_{i, k} \iff \\
%		\ncdf \inv ( 1 - p^x_{i, k}) \Vert \Smat ^{1/2} (I + \Bmat \Lmat)\t E_k \t a_{i, k} \Vert+ a_{i, k} \t E_k \bar{X} \leq \alpha_{i, k},
%\end{multline}
\begin{equation}
	\Pr (a_{i, k} \t x_k > \alpha_{i, k}) \leq p^x_{i, k} \iff
		\ncdf \inv ( 1 - p^x_{i, k}) \Vert \Smat ^{1/2} (I + \Bmat \Lmat)\t E_k \t a_{i, k} \Vert
		+ a_{i, k} \t E_k (\Amat \bar{x}_0 + \Bmat V + C + \Gmat \bar{W}) \leq \alpha_{i, k},
\end{equation}
where $\Smat^{1/2}$ denotes a matrix satisfying $\Smat = (\Smat^{1/2}) \t \Smat^{1/2}$, and where $E_k \in \Re ^{n \times (N + 1) n}$ is a matrix defined such that $E_k X = x_k$.
Similarly, for the control constraints (\ref{eq:state_chance_constraint}) one obtains
\begin{equation}
	\Pr (b_{i, k} \t u_k > \beta_{i, k}) \leq p^u_{i, k} \iff
		\ncdf \inv ( 1 - p^u_{i, k}) \Vert \Smat ^{1/2} \Lmat\t E_k ^{u \addt} b_{i, k} \Vert+ b_{i, k} \t E^u_k V \leq \beta_{i, k},
\end{equation}
where $E^u_k \in \Re ^ {m \times Nm}$ such that $E^u_k U = u_k$.

\subsection{Terminal Distribution Constraints}

The final state mean constraint is given by
\begin{equation}
	\E(x_f) = E_N ( \Amat \bar{x}_0 + \Bmat V + C + \Gmat \bar{W} ) = \bar{x}_f,
\end{equation}
which is convex in the decision variable $V$.
The final state covariance constraint is given by
%he terminal covariance constraint~(\ref{ch2:eq:terminal_covariance}) may be written as 
\begin{equation}
	E_N (I + \Bmat \Lmat) \Smat (I + \Bmat \Lmat) \t  E_N \t \leq P_f,
\end{equation}
which may be equivalently written as~\cite{Okamoto2018csl}
\begin{equation} \label{ch2:eq:convex_final_covariance}
	\Vert \Smat ^{1/2} (I + \Bmat \Lmat) \t  E_N \t P_f ^{-1/2} \Vert \leq 1,
\end{equation}
which is a convex constraint in terms of $\Lmat$.
Note that, by assumption, $P_f$ is positive definite, and hence $P_f ^{-1/2}$ exists.

\subsection{Cost Function}

The cost $J_1$ is rewritten in terms of the decision variables $\Lmat$ and $V$ as
\begin{equation}
	J_1(\Lmat, V) = \tr \big\{ \big( (I + \Bmat \Lmat) \t \Qmat (I + \Bmat \Lmat) + \Lmat \t \Rmat \Lmat \big) \Smat \big\}
	+ \Vert \Amat \bar{x}_0 + \Bmat V + C + \Gmat \bar{W} - X^d \Vert^2_{\Qmat} + V \t \bar{\Rmat} V.
\end{equation}
where $\Qmat \in \Re ^{(N + 1)n \times (N + 1) n}$ and $\Rmat \in \Re ^{Nm \times Nm}$
and $\bar{\Rmat}$
are block-diagonal matrices given by
\begin{equation}
	\Qmat = \begin{bmatrix}
		Q_0 & & \\
		& \ddots & \\
		& & Q_{N - 1} \\
		& & & 0
	\end{bmatrix}, \qquad \Rmat = \begin{bmatrix}
		R_0 & & \\
		& \ddots & \\
		& & R_{N - 1}
	\end{bmatrix},
	\qquad 
	\bar{\Rmat} = 
	\begin{bmatrix}
		\bar{R}_0 & & \\
		& \ddots & \\
		& & \bar{R}_{N - 1}
	\end{bmatrix}.
\end{equation}
Next, we consider the expression of the cost $J_2$.
Following the analysis in subsection \ref{sec:cvx_cc}, and since $\xi \t x_f$ is a Gaussian random variable, we have the relationship
\begin{equation}
	\Pr(\xi \t x_f \leq \gamma) = \ncdf \bigg( \frac{\gamma - \xi \t \E(x_f)}{\sqrt{\xi \t \Cov(x_f) \xi}} \bigg).
\end{equation}
We can thus rewrite the inequality in the cost definition (\ref{eq:J_2_def}) as
\begin{equation} \label{eq:gamma_inequality_equiv}
	\Pr(\xi \t x_f > \gamma) \leq p_f \iff
	\xi \t \E (x_f) + \sqrt{\xi \t \Cov(x_f) \xi} \ncdf\inv(1 - p_f) \leq \gamma.
\end{equation}
The minimum value $\gamma^* = J_2$ that satisfies the inequality (\ref{eq:gamma_inequality_equiv}) is obtained by setting equality in (\ref{eq:gamma_inequality_equiv}).
After substituting the decision variables $\Lmat$ and $V$ from (\ref{eq:state_mean}) and (\ref{eq:state_deviation}) into (\ref{eq:gamma_inequality_equiv}) and simplifying, we obtain the cost $J_2$ as the convex function
\begin{equation}
	J_2(\Lmat, V) = \xi \t (\Amat \bar{x}_0 + \Bmat V + C + \Gmat \bar{W})
	+ \ncdf \inv (1 - p_f) \Vert \Smat ^{1/2} (I + \Bmat \Lmat)\t E_N \t \xi \Vert.
\end{equation}

\subsection{Iterative Covariance Steering}

In the previous subsections, we have formulated the original stochastic optimal control problem as a convex optimization program with respect to a provided nominal control input.
A solution to the original, nonlinear problem can be obtained by iteratively solving the convexified problem; this procedure is, in general, referred to as successive convex programming \cite{Mao2016,Szmuk2016,Ridderhof2019nonlinear}.

First, we must introduce the following \textit{trust region} constraints that serve to restrict each successive convex problem to a domain in which the convex approximation remains valid:
\begin{equation} \label{eq:control_trust}
	\Vert \bar{u}_k - \hat{u}_k \Vert_{M^u_k} \leq \Delta^u,
\end{equation}
\begin{equation} \label{eq:state_trust}
	\Vert \bar{x}_k - \hat{x}_k \Vert_{M^x_k} \leq \Delta^x,
\end{equation}
where $M^u_k$ and $M^x_k$ are positive semi-definite weight matrices and where $\Delta^u$ and $\Delta^x$ are given deviation limits.
The subproblem to be solved, which we refer to as the covariance steering problem, is therefore given as the following convex optimization problem.
% \panos{I think $\Delta^u$ and $\Delta^x$ should also be part of the optimization problem}

%\begin{subequations} \label{eq:convex_subproblem}
%	\begin{align}
%		\underset{\Lmat, V}{\text{min}} \;\;
%		& \tr \big\{ \big( (I + \Bmat \Lmat) \t \Qmat (I + \Bmat \Lmat) + \Lmat \t \Rmat \Lmat \big) \Smat \big\} \nonumber \\
%			&\; + \Vert (\Amat \bar{x}_0 + \Bmat V + C + \Gmat \bar{W} - X^d) \Vert_{\Qmat} + V \t \bar{\Rmat} V \nonumber \\
%			&\; + \eta \big\{ \xi \t (\Amat \bar{x}_0 + \Bmat V + C + \Gmat \bar{W}) \nonumber \\
%			&\; + \ncdf \inv (1 - p_f) \Vert \Smat ^{1/2} (I + \Bmat \Lmat)\t E_N \t \xi \Vert \big\} \\
%		\text{s.t.} \;\;
%		& \ncdf \inv ( 1 - p^x_{i, k}) \Vert \Smat ^{1/2} (I + \Bmat \Lmat)\t E_k \t a_{i, k} \Vert \nonumber \\
%		&\; + a_{i, k} \t E_k (\Amat \bar{x}_0 + \Bmat V + C + \Gmat \bar{W}) \nonumber \\
%		&\; \leq \alpha_{i, k}, \quad \forall (i, k) \in \mathcal{X}, \\
%		& \ncdf \inv ( 1 - p^u_{i, k}) \Vert \Smat ^{1/2} \Lmat\t E_k ^{u \addt} b_{i, k} \Vert \nonumber \\
%		&\; + b_{i, k} \t E^u_k V\leq \beta_{i, k}, \quad \forall (i, k) \in \mathcal{U}, \\
%		& E_N ( \Amat \bar{x}_0 + \Bmat V + C + \Gmat \bar{W} ) = \bar{x}_f, \\
%		& \Vert \Smat ^{1/2} (I + \Bmat \Lmat) \t  E_N \t P_f ^{-1/2} \Vert \leq 1, \\
%		& \Vert \bar{u}_k - \hat{u}_k \Vert_{M^u_k} \leq \Delta^u, \quad \forall k \in \{1, \dots, N\}, \\
%		& \Vert \bar{x}_k - \hat{x}_k \Vert_{M^x_k} \leq \Delta^x, \quad \forall k \in \{0, \dots, N - 1\},
%	\end{align}
%\end{subequations}
\begin{subequations} \label{eq:convex_subproblem}
	\begin{align}
		\underset{\Lmat, V}{\text{min}} \;\;
		& \tr \big\{ \big( (I + \Bmat \Lmat) \t \Qmat (I + \Bmat \Lmat) + \Lmat \t \Rmat \Lmat \big) \Smat \big\} 
			 + \Vert \bar{X} - X^d \Vert^2_{\Qmat} + V \t \bar{\Rmat} V \nonumber \\ 
			 &\; + \eta \big\{ \xi \t \bar{X}
			 + \ncdf \inv (1 - p_f) \Vert \Smat ^{1/2} (I + \Bmat \Lmat)\t E_N \t \xi \Vert \big\} \\
		\text{subject to} \;\;
		& \ncdf \inv ( 1 - p^x_{i, k}) \Vert \Smat ^{1/2} (I + \Bmat \Lmat)\t E_k \t a_{i, k} \Vert 
		 + a_{i, k} \t E_k \bar{X} \leq \alpha_{i, k}, \quad \forall (i, k) \in \mathcal{X}, \\
		& \ncdf \inv ( 1 - p^u_{i, k}) \Vert \Smat ^{1/2} \Lmat\t E_k ^{u \addt} b_{i, k} \Vert
		 + b_{i, k} \t E^u_k V\leq \beta_{i, k}, \quad \forall (i, k) \in \mathcal{U}, \\
		& E_N \bar{X} = \bar{x}_f, \\
		& \Vert \Smat ^{1/2} (I + \Bmat \Lmat) \t  E_N \t P_f ^{-1/2} \Vert \leq 1, \\
		& \Vert E^u_k V - \hat{u}_k \Vert_{M^u_k} \leq \Delta^u, \quad \forall k \in \{0, \dots, N - 1\}, \\
		& \Vert E_k \bar{X} - \hat{x}_k \Vert_{M^x_k} \leq \Delta^x, \quad \forall k \in \{1, \dots, N\},
	\end{align}
\end{subequations}
where the matrix $\Smat$ is given in (\ref{eq:Smat_def}) and  the vector $\bar{X}$, which depends on $V$, is given in (\ref{eq:state_mean}).
The resulting successive convex programming algorithm is summarized in Algorithm~\ref{algo:ics}.

\begin{algorithm}
\SetAlgoLined
\KwIn{Initial state mean and covariance $\bar{x}_0$, $P_0$, initial control guess $\hat{u}$, time partition $\mathscr{P}$}
\KwOut{Control law parameters $(K_{k, \ell})$, $(v_k)$, $(\bar{x}_k)$}
	\While{termination criteria not met}{
		Propagate nominal trajectory (\ref{eq:nominal_system})\;
		Linearize (\ref{eq:linearized_system})\;
		Discretize (\ref{eq:discrete_system_integral})\;
		Calculate disturbance statistics (\ref{eq:disturbance_mean}), (\ref{eq:disturbance_covariance})\;
		Solve convex program (\ref{eq:convex_subproblem})\;
		Set control law (\ref{eq:control_law_orig})\;
		Set new nominal control $\hat{u}_k \leftarrow \bar{u}_k$\;
		}
 \caption{Iterative Covariance Steering in a Gaussian Random Field
 \label{algo:ics}}
\end{algorithm}

\section{Numerical Examples}
\label{sec:examples}

In this section the developed theory is illustrated using two examples.
The first example is a double integrator subjected to a random position-dependent external force.
The second example treats aerocapture guidance around a planet with altitude-dependent density variations.

\subsection{Double Integrator}
\label{sec:example_double_integrator}

Consider a single-dimensional double integrator with position $r$ and velocity $v$ given as unitless values.
A GRF $\Psi(r)$ acts as an external force on the system, as a function of the position, in addition to a control force $u$.
This system is described by the equations
\begin{equation}
    \begin{bmatrix}
        \dot{r} \\ \dot{v}
    \end{bmatrix}
    =
    \begin{bmatrix}
        v \\ u + \Psi(r)
    \end{bmatrix}.
\end{equation}
% \begin{subequations}
% 	\begin{equation}
% 		\dot{r} = v,
% 	\end{equation}
% 	\begin{equation}
% 		\dot{v} = u + \Psi(r).
% 	\end{equation}
% \end{subequations}
The state is normally distributed at the initial time by
\begin{equation}
	\begin{bmatrix}
		r(t_0) \\ v(t_0)
	\end{bmatrix} \sim \Ncal \left( \begin{bmatrix}
		0.1 \\ 0.1
	\end{bmatrix} , \begin{bmatrix}
		\sigma_r^2 & 0 \\
		0 & \sigma_v^2
	\end{bmatrix} \right),
\end{equation}
where $3\sigma_r = 0.05$ and $3\sigma_v = 0.01$.
The force input $\Psi$ is assumed to have zero mean and locally-periodic covariance
%$\mu(r) = 0$,
\begin{equation}
	\covfun(r, r') = \sigma_\Psi^2
	\exp \bigg( -\frac{2 \sin ^2 \big( \pi \vert r - r' \vert / p \big) }{\ell_p^2} \bigg)
	\exp \bigg( - \frac{(r - r')^2}{2 \ell_e^2} \bigg),
\end{equation}
where $\sigma_\Psi^2 = 2 \times 10 ^ {-6}$ is the variance, $p = 0.35$ is the period, $\ell_p = 0.8$ is the periodic length scale, and $\ell_e = 1$ is the exponential-quadratic length scale.
Samples of $\Psi$ are plotted in Figure~\ref{fig:di_gp_samples}.

We consider the system over the time interval $[0, 5]$ with $\mathscr{P} = (0, 1, \dots, 5)$, and so $N = 5$.
The distribution of the state at the final time is constrained by
\begin{equation}
	\bar{x}(t_f) = \begin{bmatrix}
		0.6 \\ 0.1
	\end{bmatrix}, \quad P(t_f) \leq P_f = \begin{bmatrix}
		\sigma_r^2 & 0 \\
		0 & \sigma_v^2
	\end{bmatrix} .
\end{equation}
At each step $k$, the state is constrained to lie in the region between two lines passing through the point $(0.7, 0.1)$ having slopes $\pm 0.05 / 0.1$, which is shown by dashed lines in Figure~\ref{fig:di_pos_cmp}, with a probability of at least $0.9973$.
Translating into the format (\ref{eq:state_chance_constraint}), and leveraging the subadditivity of probability, this chance constraint is represented by 
%At each step $k$, the state is constrained as in (\ref{eq:state_chance_constraint}) with
\begin{equation}
	a_{1, k} = \begin{bmatrix}
		0.212766 \\ 8.51064
	\end{bmatrix}, \quad a_{2, k} = \begin{bmatrix}
		0.30303 \\ -12.1212
	\end{bmatrix},
\end{equation}
and $\alpha_{1, k} = 1$, $\alpha_{2, k} = -1$, and $p^x_{i, k} = (1 - 0.9973) / 2$ for $i = 1, 2$.
The running control weight is $R_k = \bar{R}_k = 1$ and the state weight $Q_k$ is zero.
We only consider the quadratic cost $J_1$, and therefore we set $\eta = 0$.

Algorithm~\ref{algo:ics} was run for a single iteration from an initial guess of zero control, without explicit termination criteria, and the resulting open and closed-loop trajectories are shown in Figure~\ref{fig:di_pos_cmp}.
While for the closed-loop trajectory the confidence ellipses are not entirely within the constrained region, the chance constraints were satisfied based on 5,000 Monte Carlo trials.

% From code:

%a = [[0.212766; 8.51064], -[-0.30303; 12.1212]]
%alpha = [1, -1]
%px = (1 - 0.9973) * [1, 1] / 2

%periodic_variance = 0.02 / 100^2;
%periodic_length_scale = 0.8;
%exp_quad_periodic_length_scale = 1;
%periodic_period = 0.35;

\begin{figure}
	\centering
	% trim={<left> <lower> <right> <upper>}
	\includegraphics[trim={.0in 0.0in 0.1in 0.0in}, clip]{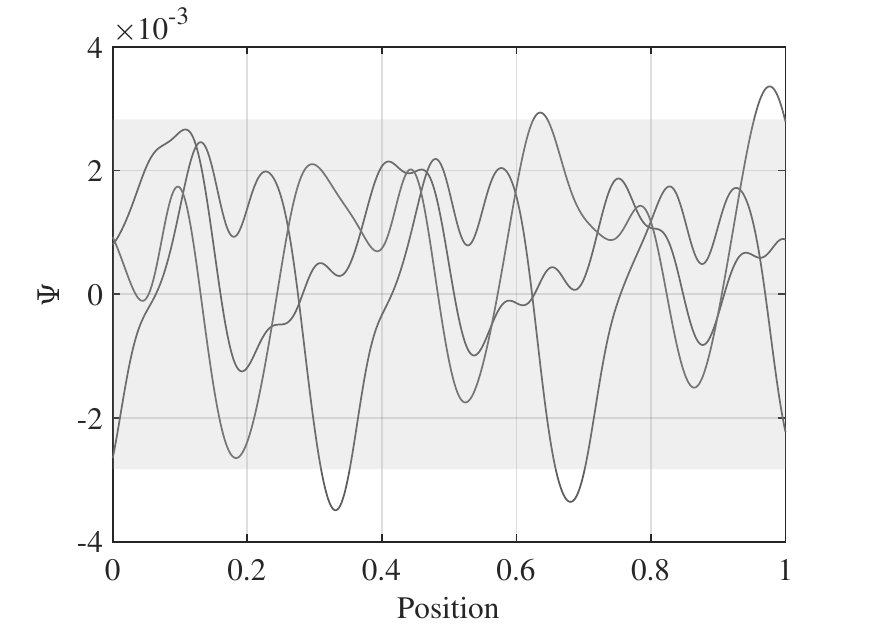}
	\caption{
		Samples of $\Psi$ with shaded $2\sigma$ confidence interval
		\label{fig:di_gp_samples}}
\end{figure}

\begin{figure}
	\centering
	% trim={<left> <lower> <right> <upper>}
	\includegraphics[trim={.0in 0.1in 0.1in 0.2in}, clip]{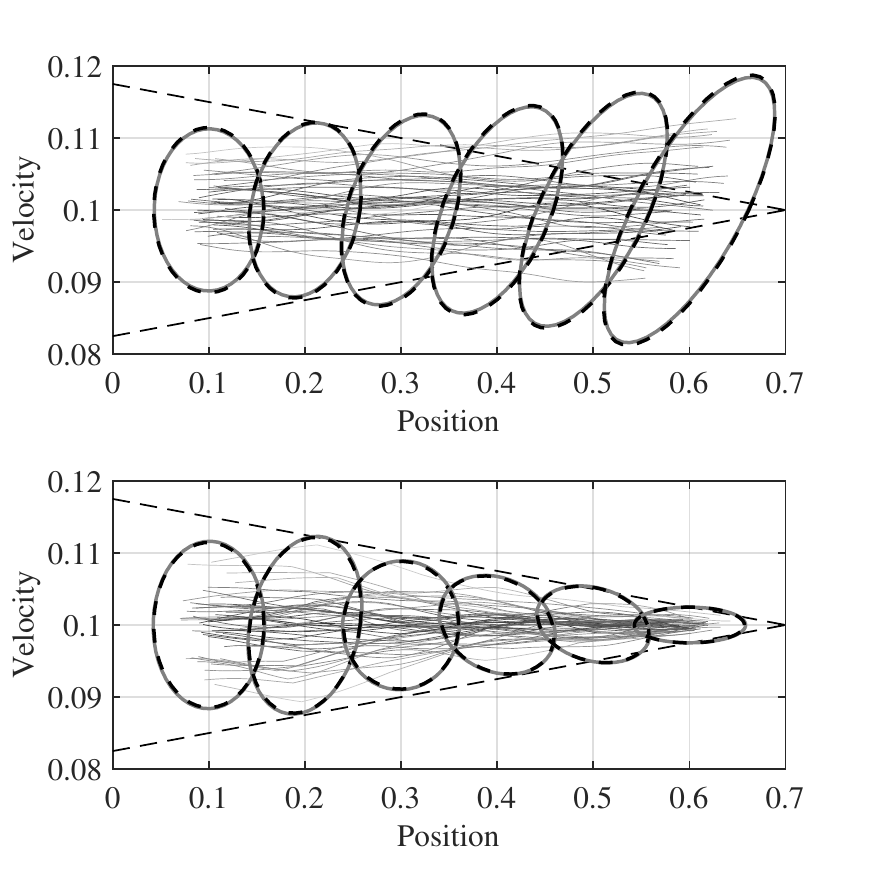}
	\caption{
		Open-loop (top) and closed-loop (bottom) trajectories of the double integrator system with 99.73\% confidence ellipses computed from linear covariance (black, dashed) and 5,000 trial Monte Carlo (gray, solid).
		\label{fig:di_pos_cmp}}
\end{figure}

\subsection{Aerocapture}
\label{sec:example_areocap}

In this subsection, we apply Algorithm~\ref{algo:ics} to the problem of aerocapture guidance, which was briefly described in Section~\ref{sec:intro}.
First, we review and motivate the aerocapture problem.

Concept studies have shown that using aerocapture in place of an all-propulsive system 
can have significant benefits in many future interplanetary space missions.
For instance, aerocapture 
can increase the delivered mass to a science orbit around Neptune by 1.4 times \cite{Lockwood2006technote,Masciarelli2004aerocapture,Lockwood2004neptune}, can decrease the required launch mass for a Mars robotic mission by 3-4 times \cite{Wright2006mars}, and can decrease the required mass for a Titan robotic mission by between 40 and 80\% \cite{Lockwood2003titan,Lockwood2006_titan_technote}.
While recent works have studied open-loop aerocapture with parametric uncertainty \cite{Heidrich2020aerocapture} and with density uncertainty modeled as a GRF \cite{Ridderhof2020aeroconf,Albert2021}, treating closed-loop aerocapture with uncertainty remains an open problem.

\subsubsection{Mission Design}

The aerocapture mission profile is shown in Figure~\ref{fig:aerocapture_mission}.
Following atmospheric flight, the vehicle will perform a periapsis raising burn (to raise the periapsis out of the planet's atmosphere) followed by an apoapsis clean up burn.
Both the final orbit and the $\Delta v$ cost are determined by the target periapsis and target apoapsis; the $\Delta v$ cost is also determined by the vehicle state following the atmospheric flight segment.

\begin{figure}
	\centering
	\includegraphics{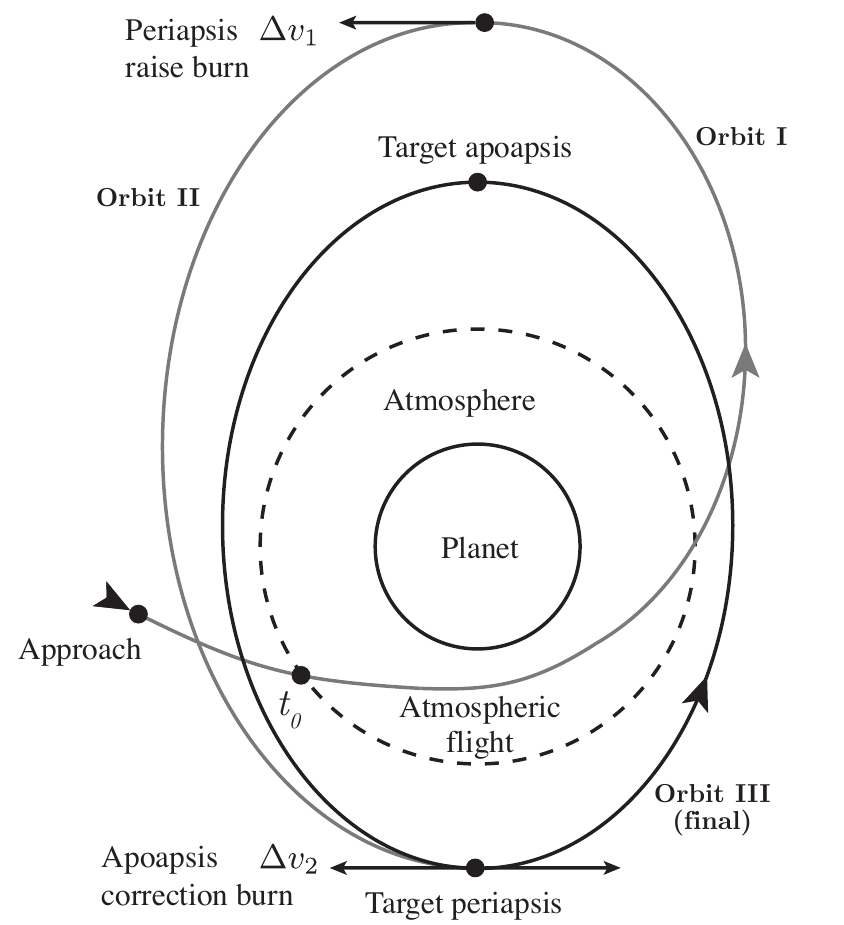}
	\caption{
		Aerocapture mission overview
		\label{fig:aerocapture_mission}}
\end{figure}

Let $r$, $v$, and $\gamma$ be the vehicle radius, planet-relative velocity, and planet-relative flight path angle (FPA).
The apoapsis radius of the orbit following atmospheric flight is a function of the state $x_f = (r_f, v_f, \gamma_f)$ at the final time, given by
\begin{equation}\label{eq:ra_ex_value}
	r_{a,\text{ex}} = a_{\text{ex}} \bigg(1 + \sqrt{ 1 - \frac{r_f^2 v_f^2 \cos ^ 2 \gamma_f}{\mu_{\text{grav}} a_{\text{ex}}}} \bigg),
\end{equation}
where $\mu_{\text{grav}}$ is the planet's gravitational parameter, and where $a_{\text{ex}}$ is the semi-major axis at atmospheric exit given by
\begin{equation}
	a_{\text{ex}} = \frac{\mu_{\text{grav}}}{2 \mu_{\text{grav}} / r_f - v_f ^ 2}.
\end{equation}

Following the atmospheric flight segment, the vehicle coasts to its apoapsis, where it has velocity
%The velocity at the first time of apoapsis following atmospheric flight is
\begin{equation}
	v_{a_1}^- = \sqrt{ v_f^2 + 2 \mu_{\text{grav}} \bigg( \frac{1}{r_{a, \text{ex}}} - \frac{1}{r_f} \bigg) }.
\end{equation}
However, the required velocity at the radius $r_{a, \text{ex}}$ for the periapsis to equal to the desired periapsis $r_{p, \text{targ}}$ is
\begin{equation}
	v_{a_1}^+ = \sqrt{ 2 \mu_{\text{grav}} \bigg( \frac{1}{r_{a, \text{ex}}} - \frac{1}{r_{a, \text{ex}} + r_{p,\text{targ}}} \bigg) }.
\end{equation}
The first impulsive maneuver increases the velocity from $v_{a_1}^-$ to $v_{a_1}^+$, and hence
\begin{equation}
	\Delta v_1 = v_{a_1}^+ - v_{a_1}^-.
\end{equation}
Next, the vehicle coasts to the periapsis $r_{p, \text{targ}}$, where it has velocity
\begin{equation}
	v_{p_1}^- = \sqrt{2 \mu_{\text{grav}} \bigg( \frac{1}{r_{p, \text{targ}}} - \frac{1}{r_{a, \text{ex}} + r_{p, \text{targ}}} \bigg) },
\end{equation}
whereas the velocity at this point required for the apoapsis to be equal to the target apoapsis $r_{a, \text{targ}}$ is
\begin{equation}
	v_{p_1}^+ = \sqrt{2 \mu_{\text{grav}} \bigg( \frac{1}{r_{p, \text{targ}}} - \frac{1}{r_{a, \text{targ}} + r_{p, \text{targ}}} \bigg) }.
\end{equation}
The second impulsive maneuver corrects the discrepancy in the velocity at periapsis, and thus
\begin{equation}
	\Delta v_2 = \vert v_{p_1}^+ - v_{p_1}^- \vert.
\end{equation}
The total fuel cost is the sum
\begin{equation} \label{eq:delta_v_value}
	\Delta v = \Delta v_1 + \Delta v_2.
\end{equation}

\subsubsection{Atmospheric Flight}

During atmospheric flight, which is described in Figure~\ref{fig:atmo_flight_coordinates}, a vehicle flying at a trimmed angle of attack can steer by banking the lift vector about the velocity vector, as shown in Figure~\ref{fig:bank_angle_control2_aoa}.
The vertical component in the lift vector is set via the cosine of the bank angle, and the sign of the bank angle is set for lateral control.
In this example, we only consider the longitudinal guidance, and so the control input during atmospheric flight is the bank angle cosine.
For a more general treatment, a separate lateral channel could be added to determine the bank sign \cite{Lu2015aerocapture}, and the covariance function could be modified to be a function of longitude and latitude in addition to altitude.

\begin{figure}
	\centering
	\includegraphics{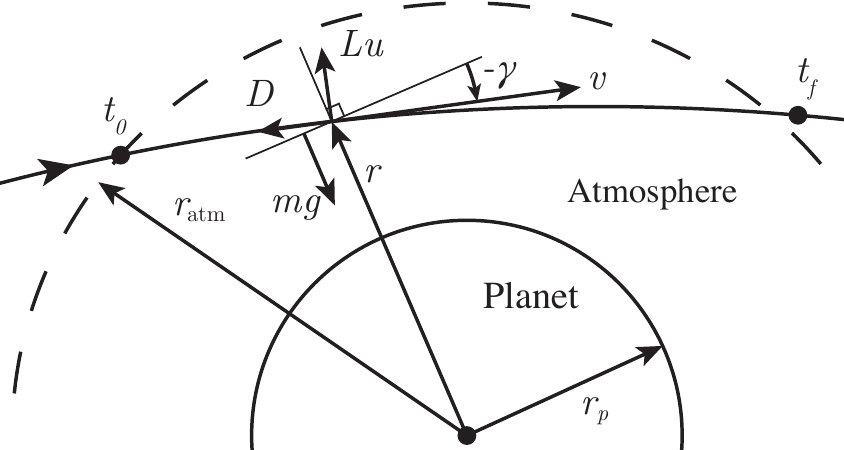}
	\caption{
		Atmospheric flight coordinates with lift $L$ and drag $D$
		\label{fig:atmo_flight_coordinates}}
\end{figure}

\begin{figure}
	\centering
	\includegraphics{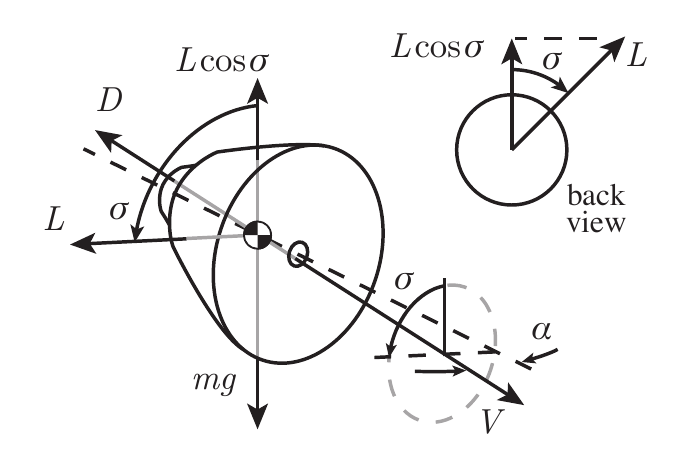}
	\caption{
		Bank angle control with bank angle $\sigma$ and angle of attack $\alpha$
		\label{fig:bank_angle_control2_aoa}}
\end{figure}

The vehicle dynamics during atmospheric flight are described by the system of equations
\begin{subequations} \label{eq:aerocap_eom}
	\begin{align}
		\dot{r} &= v \sin \gamma, \\
		\dot{v} &= - \frac{\rho(r) v^2}{2 B_c} - \frac{\mu_{\text{grav}} \sin \gamma}{r^2}, \\
		\dot{\gamma} &= \frac{\rho(r) v (L / D)}{2 B_c} u - \bigg( \frac{\mu_{\text{grav}}}{r^2} - \frac{v^2}{r} \bigg) \frac{\cos \gamma}{v},
	\end{align}
\end{subequations}
where the input $u$ is the cosine of the bank angle, $\rho$ is atmospheric density, $L/D$ is the lift-to-drag ratio, and $B_c = m / S C_D$ is the spacecraft ballistic coefficient in terms of mass $m$, reference area $S$, and drag coefficient $C_D$.
The ballistic coefficient and lift-to-drag ratio are set to $B_c = 150$ \si{kg/m^2} and $L/D = 0.2$.
The planet is Mars, which is modeled as a sphere of radius \mbox{$r_p = 3397$ \si{km}} and gravitational parameter \mbox{$\mu_{\text{grav}} = \num{4.2828e13}$ \si{m^3/s^2}}.
At the initial time, the state has mean
%\begin{equation}
%	\bar{x}_0 =
%%	\begin{bmatrix}
%%		125 \, \si{km} + r_p \\
%%		6.1 \, \si{km/s} \\
%%		-10 \, \si{deg}
%%	\end{bmatrix}.
%	\begin{bmatrix}
%		\bar{r}_0 \\
%		\bar{v}_0 \\
%		\bar{\gamma}_0
%	\end{bmatrix}.
%\end{equation}
$\bar{x}_0 = ( \bar{r}_0, \bar{v}_0, \bar{\gamma}_0)$, with $\bar{r}_0 = 125 \, \si{km} + r_p$, $\bar{v}_0 = 6.1 \, \si{km/s}$, and
%$\bar{\gamma}_0 = -10 \, \si{deg}$.
$\bar{\gamma}_0 = -10.0128 ^ \circ$.
The initial flight path angle is set so that a constant control input $u \equiv 0$ results in the apoapsis after atmospheric exit $r_{a,\text{ex}}$ being equal to the target apoapsis $r_{a,\text{targ}}$.
While the proposed method allows for the initial state to be Gaussian distributed, for this example we set the initial state covariance to be zero so that the effect of the atmospheric disturbances is more clear.

The atmospheric density is given by
\begin{equation}
	\rho = \bar{\rho} (1 + \delta \rho),
\end{equation}
where $\bar{\rho}$ is a known, smooth function describing the nominal density.
The \textit{density variation} $\delta \rho$ is a zero-mean GRF taking values as a function of the altitude $h = r - r_p$, where $r_p$ is the planet radius.
Based on the MarsGRAM atmosphere model \cite{Justus2002}, we define the density variation covariance function as
\begin{equation}
	\covfun (h_1, h_2) = \exp \bigg(-\frac{ \vert h_1 - h_2 \vert }{H_{\text{scale}}} \bigg)
	\times \begin{cases}
		b\big( {\mathrm{min}}(h_1, h_2) \big), & {\mathrm{min}}(h_1, h_2) < h_\text{trans}, 
			\vspace{3pt}\\
		\sigma_{\rho, \text{max}} ^ 2, & {\mathrm{min}}(h_1, h_2) \geq h_\text{trans},
	\end{cases}
\end{equation}
where $H_{\text{scale}}$ is the scale height, and where
\begin{equation}
	b(h) = \sigma_{\rho, \text{max}} ^ 2 \exp \bigg( \frac{h - h_\text{trans}}{c_\text{scale}} \bigg).
\end{equation}
The constants $h_\text{trans}$ and $c_\text{scale}$ determine the scale of the exponential variance model, and $\sigma_{\rho, \text{max}} ^ 2$ is the maximum density variance, which is realized for altitudes $h \geq h_\text{trans}$.
We use the values $H_\text{scale} = 11.1$ \si{km}, $\sigma_{\rho, \text{max}}^2 = 1480$ \si{(kg/m^3)^2}, $h_\text{trans} = 120$ \si{km}, and $c_\text{scale} = 20$ \si{km}.
The nominal density and samples of $\delta \rho$ are shown in Figures~\ref{fig:atmo_nom} and \ref{fig:atmo}.
The nominal density $\bar{\rho}(h)$ is provided by MarsGRAM \cite{Justus2002}.
We remark that while in this example the atmosphere is taken as a function of the altitude, more general models including longitude and latitude dependence could also be used, provided an appropriate covariance function.

Without loss of generality, we let $t_0 = 0$.
The final time is set to $t_f = 400$ \si{s} and \mbox{$\mathscr{P} = (0, 50, 75, \dots, 425, 450, 400)$ \si{s}.}

\begin{figure}
	\centering
	% trim={<left> <lower> <right> <upper>}
	\includegraphics[trim={.0in 0.0in 0.1in 0.0in}, clip]{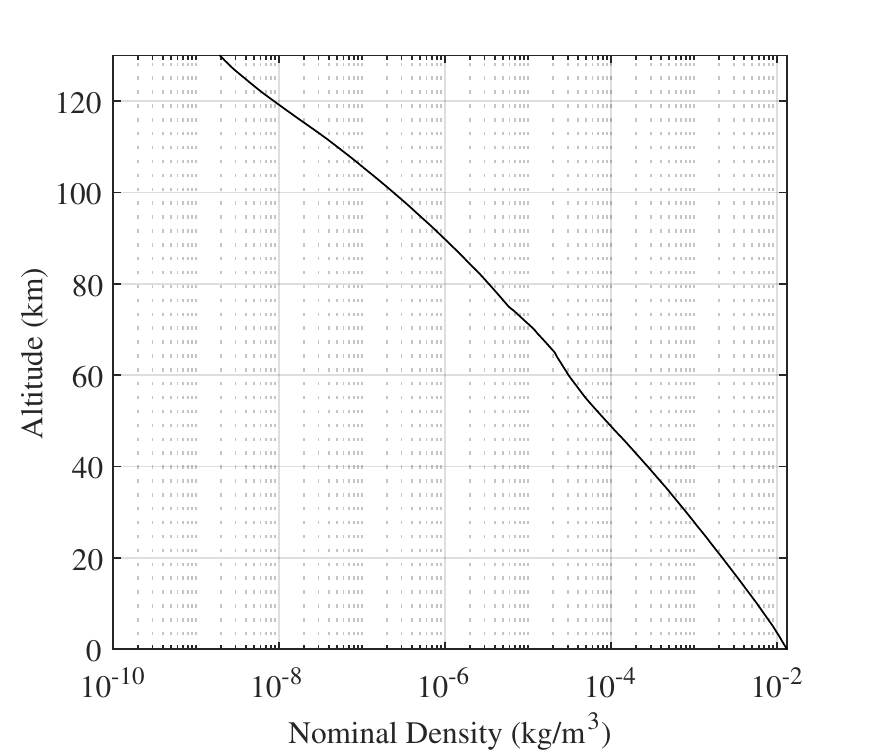}
	\caption{
		Nominal density profile
		\label{fig:atmo_nom}}
\end{figure}

\begin{figure}
	\centering
	% trim={<left> <lower> <right> <upper>}
	\includegraphics[trim={.0in 0.0in 0.1in 0.0in}, clip]{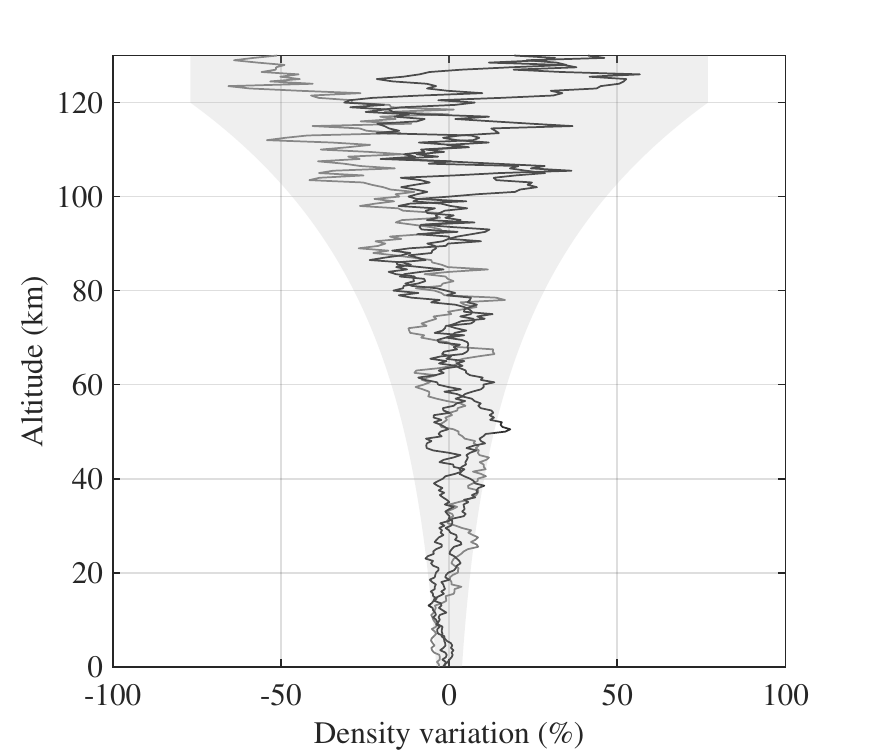}
	\caption{
		Samples from the density variation process with the $2\sigma$ confidence interval shaded
		\label{fig:atmo}}
\end{figure}

% Taken from the code:
%Bc_si = 150            % Ballistic coefficient
%r_p_si = 3397e3        % Mars radius, m
%mu_si = 4.282837e13    % Mars grav. parameter m3/s2
%target_periapsis_si = (250 + 3397) * 1e3
%target_apoapsis_si = (33786 + 3397) * 1e3
%apoapsis_after_exit_max_in_planet_radii = 8.5
%init_state_si = [125e3 + 3397e3; 6.1e3; -10 * pi / 180; 0];

% Normalization:
%du by r_p
%tu by:
%function t = tu_to_seconds(obj, tnd)
%t = tnd / sqrt(obj.g0_si / obj.r_p_si);
%end

\subsubsection{Feedback Control Design}

The bank angle control during atmospheric flight is determined to minimize the $\Delta v$ required to reach the target orbit apoapsis $r_{a, \text{targ}} = 5 r_p$ and periapsis $r_{p, \text{targ}} = 2 r_p$.
Since in the stochastic setting $\Delta v$ is a random variable, we are able to explicitly minimize the 99\textsuperscript{th} percentile of the total $\Delta v$ cost, rather than simply minimizing the expected $\Delta v$ cost.
To this end, we approximate $\Delta v$ from (\ref{eq:delta_v_value}) to first order as
\begin{equation}
	\Delta v (x_f) \approx \Delta v(\hat{x}_f) + \frac{\partial \Delta v}{\partial x_f} \bigg\vert_{\hat{x}_f} (x_f - \hat{x}_f),
\end{equation}
where $\hat{x}_f$ is obtained from integrating the nominal dynamics (\ref{eq:nominal_system}), and set
\begin{equation} \label{eq:xi_for_dv}
	\xi \t = \frac{\partial \Delta v}{\partial x_f}.
\end{equation}
The cost $J_2$ as in (\ref{eq:J_2_def}) with $\xi$ as in (\ref{eq:xi_for_dv}) and with $p_f = 0.1$ is thus approximately equal to the  99\textsuperscript{th} percentile of $\Delta v$.
Since, in this case, the final state mean and covariance are included in the cost function, we do not enforce the final state constraints (\ref{eq:final_state_distribution_constraints}).

%The unconstrained minimization of $\Delta v$ can lead to a high probability of failure to capture (i.e., the orbit remains hyperbolic after atmospheric flight) \cite{Ridderhof2020aeroconf}.
%We thus constrain, with probability at least $0.9973$, the apoapsis radius immediately following atmospheric exit $r_{a, \text{ex}}$ to be less than a maximum value $r_{a, \text{ex,max}} = 8.5$ planet radii.
%This constraint is encoded as a half-plane chance constraint of the form (\ref{eq:state_chance_constraint}).
%First, we approximate $r_{a, \text{ex}}(x_f)$ from (\ref{eq:ra_ex_value}) to first order as
%\begin{equation}
%	r_{a, \text{ex}}(x_f) \approx r_{a, \text{ex}}(\hat{x}_f) + \frac{\partial r_{a, \text{ex}}}{\partial x_f} \bigg\vert_{\hat{x}_f} (x_f - \hat{x}_f).
%\end{equation}
%We then obtain the state chance constraint at the final time as
%\begin{equation} \label{eq:apoapsis_chance_constraint}
%	\Pr \big( r_{a, \text{ex}}(x_f) \leq r_{a, \text{ex,max}} \big)
%	\approx \Pr \bigg( \underbracket{ \frac{\partial r_{a, \text{ex}}}{\partial x_f} \bigg\vert_{\hat{x}_f}}_{ = a\t_{1, N}} x_f >
%	\underbracket{ r_{a, \text{ex,max}} - r_{a, \text{ex}}(\hat{x}_f) + \frac{\partial r_{a, \text{ex}}}{\partial x_f} \bigg\vert_{\hat{x}_f} \hat{x}_f }_{= \alpha_{1, N}} \bigg)
%\end{equation}

Leveraging the subadditivity of probability, we constrain the probability that $u_k \in [-1, +1]$ to be at least $0.9973$ by enforcing the constraints
\begin{equation}
	\Pr (u_k \leq 1) \geq 1 - p^u / 2, \;\; \text{and} \;\; \Pr (u_k \geq -1) \geq 1 - p^u / 2,
\end{equation}
for $p^u = 1 - 0.9973$, and for $k = 0, \dots, N - 1$.

The desired trajectory is set as $x^d_k = \bar{x}_k$ so that the state-error penalty $Q_k$ penalizes the running state covariance.
In particular, we penalize variations in the dynamic pressure $q = \rho v^2/2$, since excessive deviation from the nominal lift and drag forces will invalidate the linear approximation of the dynamics.
We thus set
\begin{equation} \label{eq:aerocap_Qk}
	Q_k = \hat{q} ^{-2} \bigg( \frac{\partial q}{\partial x} \bigg) \t \bigg( \frac{\partial q}{\partial x} \bigg),
\end{equation}
where $\hat{q}$ is the dynamic pressure along the nominal trajectory, and where the terms on the right-hand side of (\ref{eq:aerocap_Qk}) are evaluated at $\hat{x}_k$.
The running control weights are set to $R_k = 2 \times 10 ^ {-2}$ and $\bar{R}_k = 0$ for each step $k$, and $\eta = 1$.
The first and final time steps are taken to be longer than the intermediate steps to improve computational performance, since smaller time steps in these regions was observed to not be beneficial.

Finally, the change in the mean control for each iteration was limited as in (\ref{eq:control_trust}) with $\Delta_u = 0.1$ and $M_k^u = 1$; the change in the mean final state was constrained as in (\ref{eq:state_trust}) with $\Delta_x = 0.1 r_p$, $M^x_{N} = (\partial r_{a,\text{ex}}/ \partial x)\t (\partial r_{a,\text{ex}}/ \partial x)$, and $M^x_k = 0$ for $k = 1, \dots, N - 1$.

%$\Delta_u = 0.1$
%$M^u_k = 1$
%
%$\Delta_x = 0.5$ at final step
\subsubsection{Results}

Algorithm~\ref{algo:ics} was run for three iterations, starting with the initial control guess $\hat{u}_k = 0$ for all $k = 0, \dots, N - 1$.
The number of iterations was fixed, and the termination of the algorithm was determined by user feedback.
The nominal aerocapture trajectory resulting from both the initial guess and from the final nominal control are shown in Figure~\ref{fig:ac_nominal}.
The resulting probability distributions of $\Delta v$ following each iteration, including the initial open-loop guess, were computed by both the linear covariance approximation and by a 5,000 trial Monte Carlo, and are plotted in Figure~\ref{fig:dv_pdf_iter}.
First, we note that the linear covariance approximation (plotted as a PDF) reasonably approximates the empirical distribution (shown as a histogram) obtained from Monte Carlo.
One source of error between the linear covariance and the Monte Carlo distributions follows from the absolute value in the $\Delta v$ cost corresponding to the apoapsis cleanup burn.
Regardless, as shown by Figure~\ref{fig:dv_pdf_iter}, the linear covariance approximation serves as a useful surrogate for the optimization.
Despite, for example, the mismatch of the linear covariance probability density in Figure~\ref{fig:dv_pdf_iter}(c), the Monte Carlo distribution is consistently shifted and shaped in each iteration to have a lower upper percentile cost.
Using the final control law, the 99\textsuperscript{th} percentile of $\Delta v$ from the 5,000 Monte Carlo trails was 314 \si{m/s}, whereas to the open loop 99\textsuperscript{th} percentile was 717 \si{m/s}.

Next, consider the control inputs for each iteration, shown in Figure~\ref{fig:control_iter}.
With progressive iterations, the nominal control tends to increase the vertical lift in the first part of the trajectory while decreasing the lift in the final part of the trajectory.
Around the maximum dynamic pressure, which occurs nominally at 147 \si{s}, the nominal vertical lift is set to almost zero by the final iteration, which allows for the feedback control to have higher variance while ensuring that the control remains between $\pm1$ with high probability.

\begin{figure}
	\centering
	% trim={<left> <lower> <right> <upper>}
	\includegraphics[trim={0in 0.15in 0.1in 0.25in}, clip]{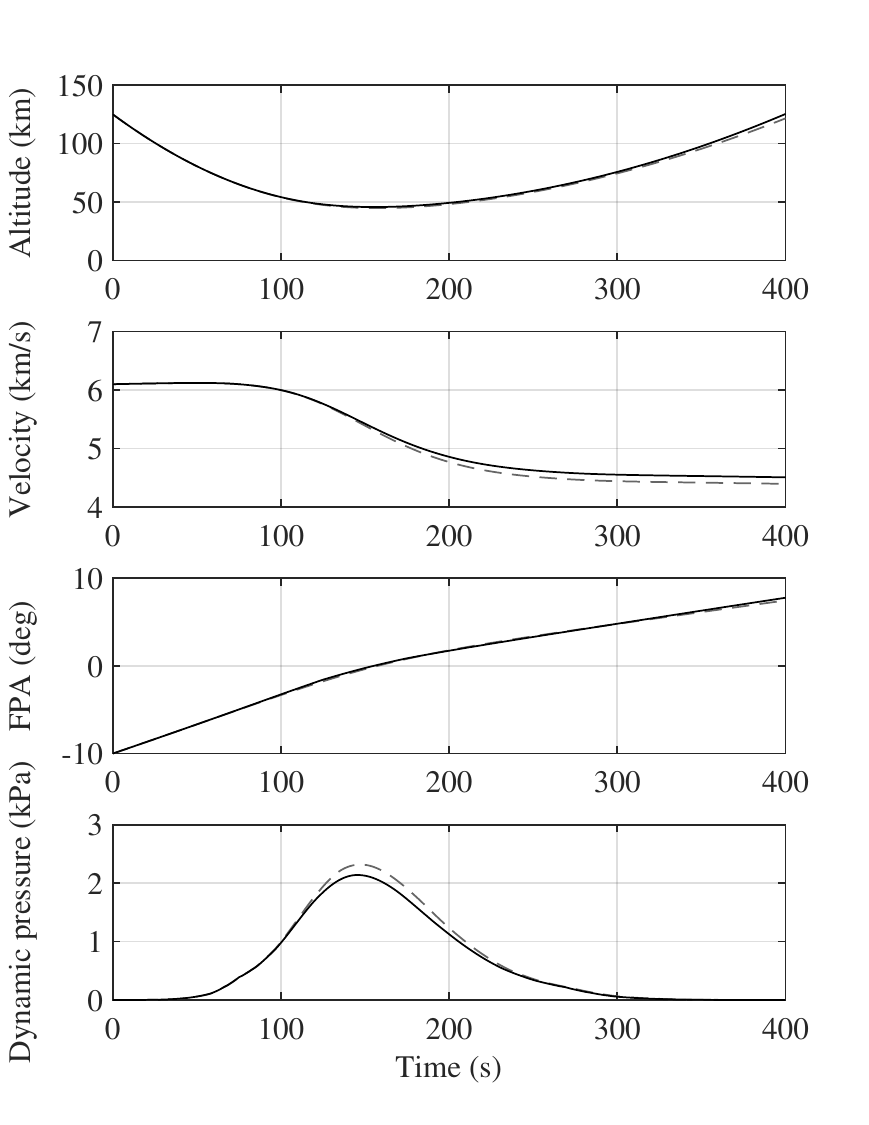}
	\caption{
		Nominal aerocapture trajectories for the initial control guess (dashed) and the final iteration (solid).
		\label{fig:ac_nominal}}
\end{figure}

Sample Monte Carlo state error and control trajectories from the initial open-loop guess and from the final optimized closed-loop trajectory are shown in Figures~\ref{fig:guess_traj} and \ref{fig:soln_traj}.
In both of these figures, the $\pm3\sigma$ bounds on the states computed from Monte Carlo are shown as dashed lines; in the control plots the dashed lines show the control limits, and stars mark the mean and $\pm3\sigma$ bounds on the control computed from Monte Carlo.
Note that the control inputs were saturated, despite the $\pm3\sigma$ limits lying slightly outside the input bounds.
At each discrete time step $t_k$, the $\pm3\sigma$ bounds computed by the linear covariance approximation are shown by error bars.
Interestingly, the linear covariance approximation is more accurate for the guess trajectory shown in Figure~\ref{fig:guess_traj}, similar to the improved accuracy of the $\Delta v$ approximation for the guess trajectory in Figure~\ref{fig:dv_pdf_iter}.
The decreased approximation accuracy along the optimized trajectory is a consequence of the closed-loop control, since, as shown in Figure~\ref{fig:dynp_compare}, deviations in dynamic pressure from the nominal trajectory increase in response to corrective controls.

Finally, we return to the approximation of the density statistics by evaluating the covariance function along the nominal trajectory, which is the fundamental assumption used to represent the spatially-defined density uncertainty as a temporal random process.
Samples of the density variation along Monte Carlo sample trajectories are plotted in Figure~\ref{fig:delrho_cov_cmp} along with $\pm3\sigma$ bounds computed by both Monte Carlo and from the approximate covariance function $\hat{\covfun}(t, t)$ as in (\ref{eq:GRF_hat}).
The close agreement between the Monte Carlo covariance and the approximate covariance $\hat{\covfun}$ suggests that, at least for the present aerocapture problem, taking the random field statistics to be a function of time along the nominal trajectory is a good approximation.

\begin{figure*}
	\centerfloat
	% trim={<left> <lower> <right> <upper>}
	\includegraphics[trim={.5in 0.0in 0.5in 0.0in}, clip]{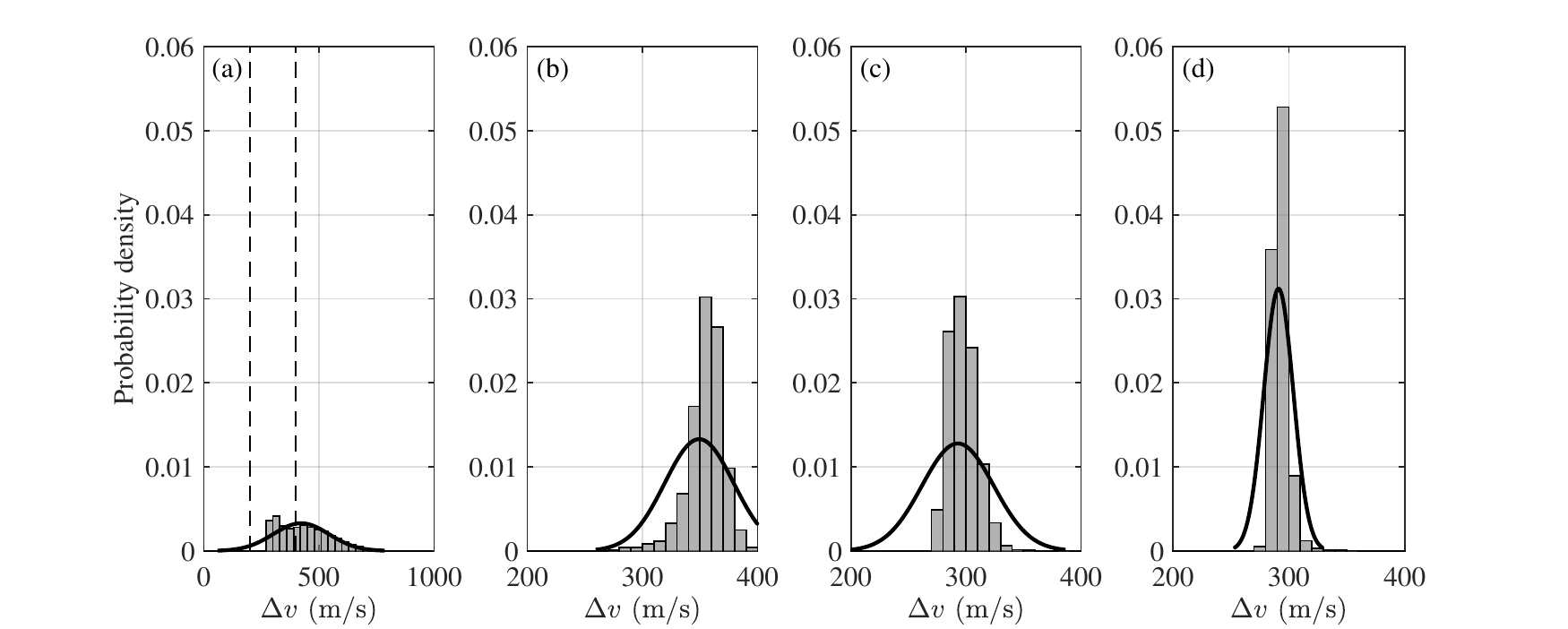}
	\caption{
		$\Delta v$ probability density for each iteration of Algorithm~\ref{algo:ics}.
		Note that plots (b--d) only show from 200 to 400 \si{m/s}, which is the interval between the dashed lines in plot (a).
		\label{fig:dv_pdf_iter}}
\end{figure*}

\begin{figure*}
	\centerfloat
	% trim={<left> <lower> <right> <upper>}
	\includegraphics[trim={.5in 0.0in 0.5in 0.1in}, clip]{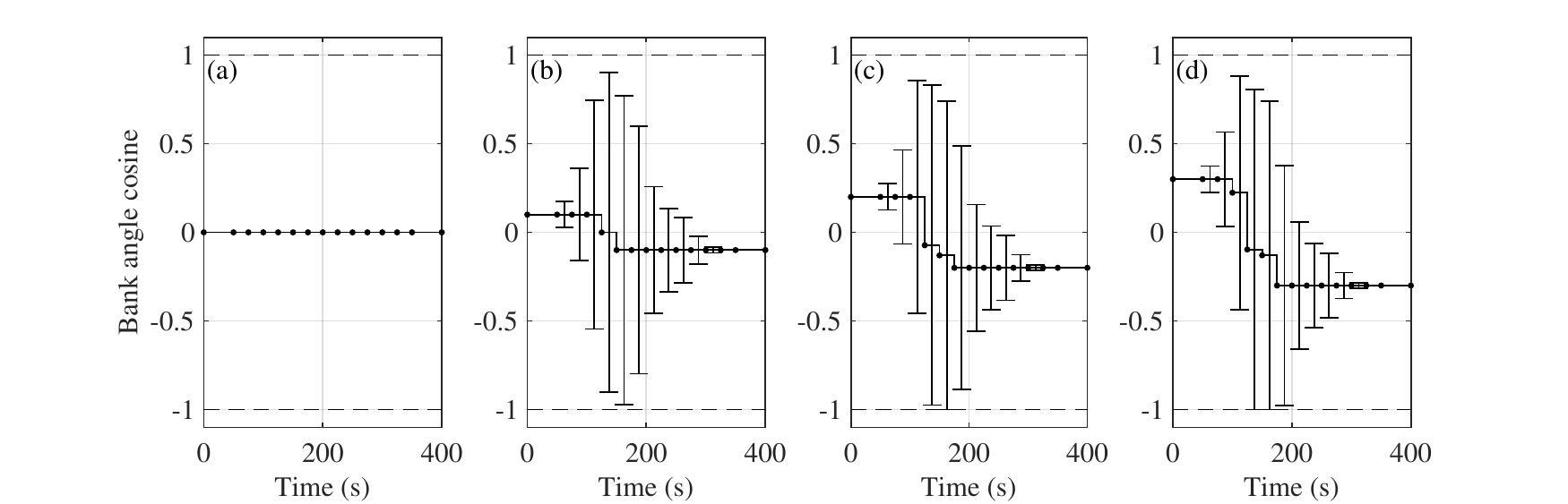}
	\caption{
		Control inputs with $\pm 3\sigma$ confidence intervals computed from linear covariance for each iteration of Algorithm~\ref{algo:ics}.
		\label{fig:control_iter}}
\end{figure*}

\begin{figure}
    \centering
	% trim={<left> <lower> <right> <upper>}
    \includegraphics[trim={0in 0.45in 0.1in 0.45in}, clip]{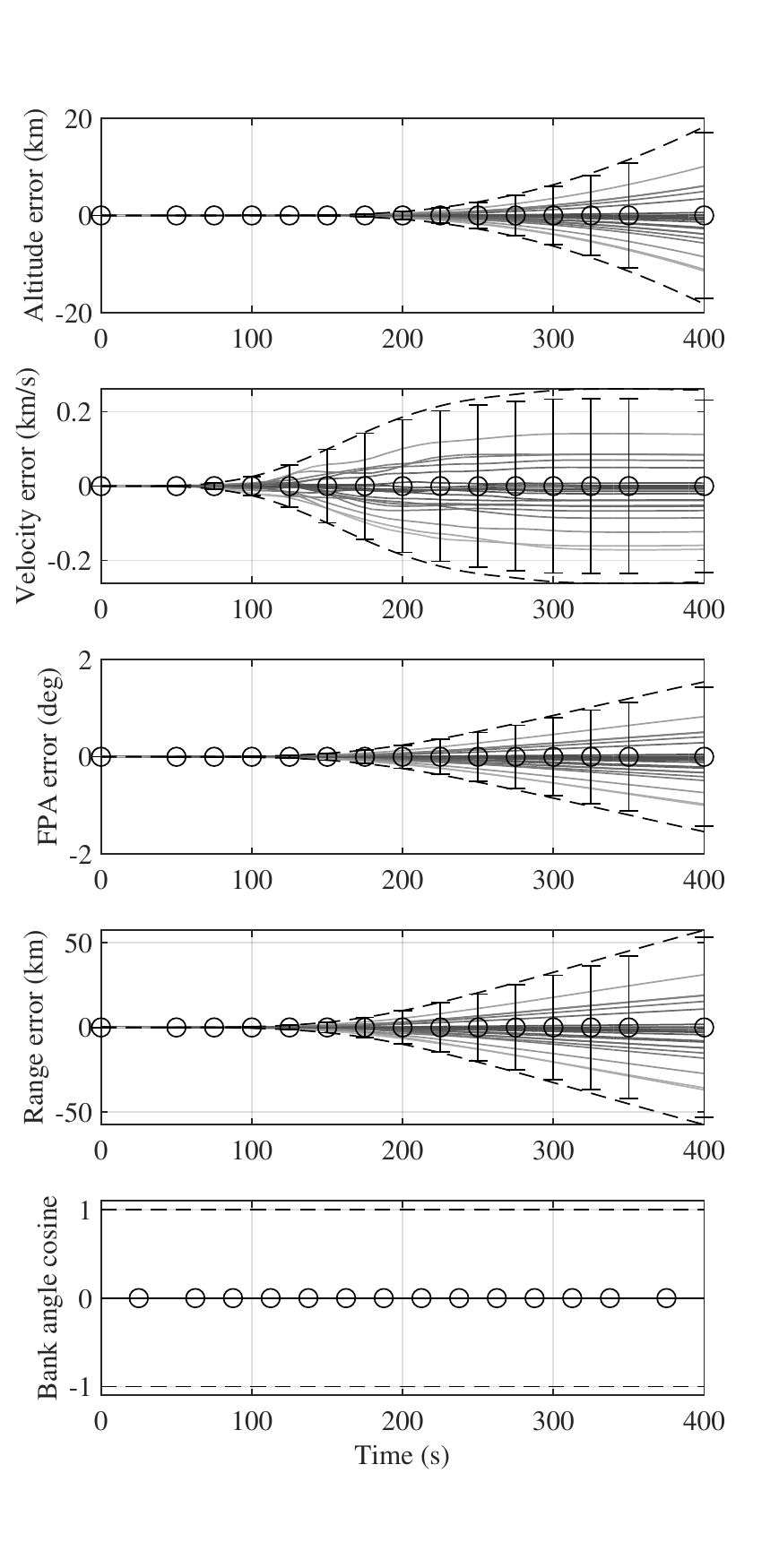}
    \caption{State error trajectories and control inputs for the initial open-loop trajectory}
    \label{fig:guess_traj}
\end{figure}

\begin{figure}
    \centering
	% trim={<left> <lower> <right> <upper>}
    \includegraphics[trim={0in 0.45in 0.1in 0.45in}, clip]{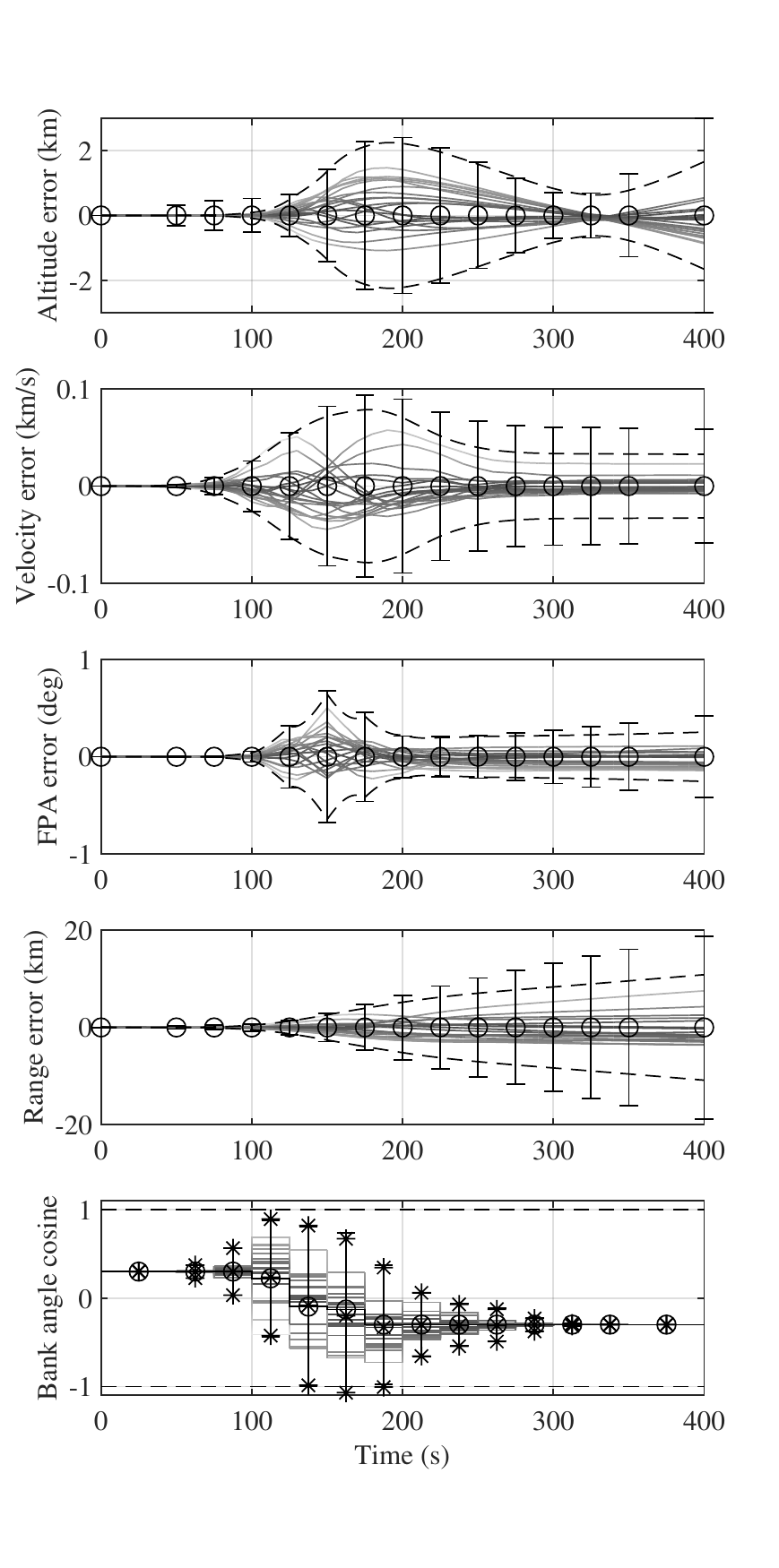}
    \caption{State error trajectories and control inputs for the optimized closed-loop trajectory}
    \label{fig:soln_traj}
\end{figure}

\begin{figure}
    \centering
	% trim={<left> <lower> <right> <upper>}
    \includegraphics[trim={0in 0.0in 0.1in 0.0in}, clip]{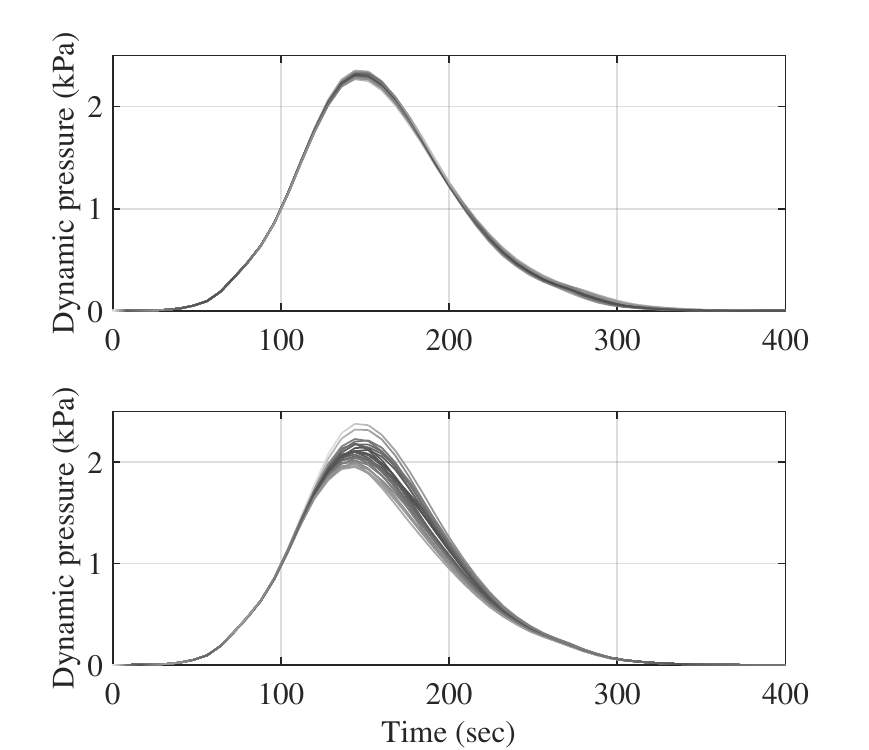}
    \caption{Dynamic pressure Monte Carlo sample trajectories from the initial open-loop trajectory (top) and from the optimized closed-loop trajectory (below)}
    \label{fig:dynp_compare}
\end{figure}

\begin{figure}
    \centering
	% trim={<left> <lower> <right> <upper>}
    \includegraphics[trim={0in 0.0in 0.1in 0.0in}, clip]{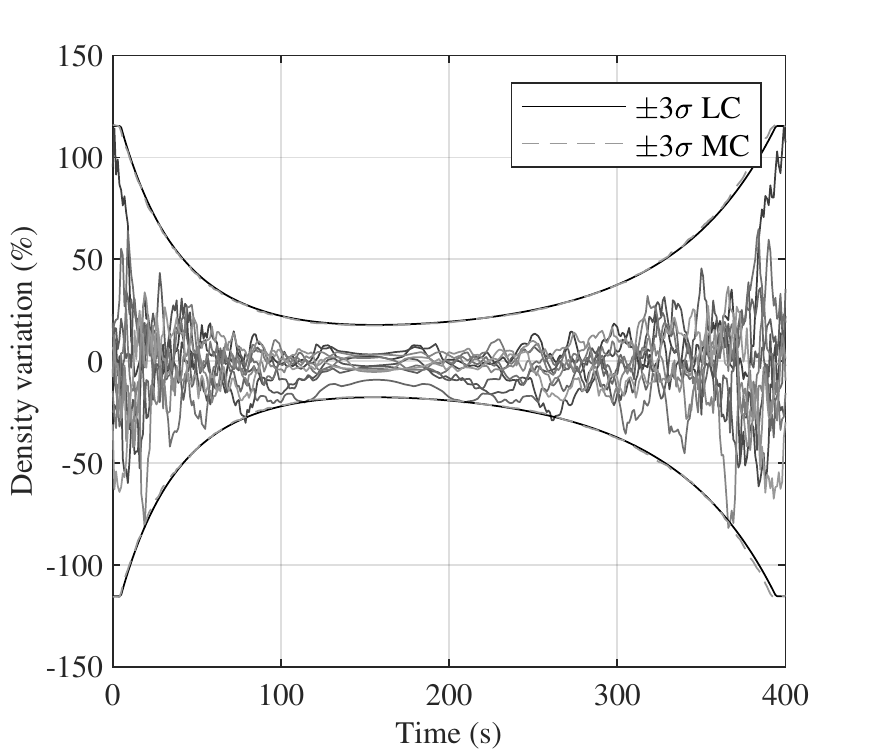}
    \caption{Density variation values along Monte Carlo trajectories with $\pm3\sigma$ bounds computed from Monte Carlo (MC) and from the linear covariance (LC) approximation}
    \label{fig:delrho_cov_cmp}
\end{figure}

\section{Conclusion}
\label{sec:conclusion}

In this paper, a method is presented for chance-constrained stochastic control of systems subjected to a spatially-dependent uncertainty modeled as a GRF.
Along a fixed nominal trajectory, spatially-dependent uncertainty becomes time-dependent, and accordingly, spatial correlations are approximated as temporal correlations.
An integral equation is derived to compute the temporal correlations of random disturbances on a dynamical system due to a GRF.
Following a linear approximation of the system dynamics, the joint optimization of the nominal and feedback controls is derived as a convex program.
The solution to the original stochastic optimal control problem is obtained by successively performing convex optimization with respect to the linearized system.
The proposed method was demonstrated on both a simple double integrator example and on a realistic aerocapture problem.
In future work, the proposed method can be applied to problems with more sophisticated disturbance models, such as aerocapture or hypersonic vehicle guidance with a three-dimensional atmosphere model.

\section*{Funding Sources}

This work was supported by NASA Space Technology Research Fellowship award 80NSSC17K0093.

%\section*{Acknowledgment}
\bibliography{gpcs21}

\end{document}